\documentclass[pdflatex,sn-mathphys-num]{sn-jnl}
\usepackage{amsmath,amssymb,amsfonts,amsthm}
\usepackage{mathtools,thm-restate,thmtools}
\usepackage{latexsym,bbm,xspace,graphicx,float}
\usepackage{physics}                            \usepackage[ruled,vlined,hangingcomment,longend,linesnumbered,resetcount]{algorithm2e}
   
\SetCommentSty{mycommfont}
\usepackage[svgnames,dvipsnames]{xcolor}
\definecolor{red}{HTML}{E51400}  
\definecolor{blue}{HTML}{0050EF} 
\definecolor{purple}{HTML}{AA00FF}
\definecolor{green}{HTML}{006265}
\definecolor{brown}{HTML}{AC4022}
\definecolor{lightlavender}{HTML}{f0edf5}
\definecolor{lavender}{HTML}{ab9dbd}
\definecolor{lightgray}{HTML}{f6f6f6}
\definecolor{gray}{HTML}{ececec}
\definecolor{lightyellow}{HTML}{fcf4dd}
\definecolor{yellow}{HTML}{f3ce85}
\definecolor{remark}{HTML}{f6809f}
\definecolor{magenta}{HTML}{7f007f}
\definecolor{dsblue}{HTML}{814f61}

\usepackage{nicefrac}
\usepackage{booktabs,tabularx}
\usepackage{graphicx}
\usepackage[strings]{underscore}
\usepackage{nicematrix}
\usepackage{multirow,multicol}

\usepackage[capitalise,nameinlink]{cleveref}
\usepackage{nameref}
\usepackage[acronym]{glossaries-extra}
\usepackage{xspace}
\setabbreviationstyle[acronym]{long-short}   
\usepackage{tikz}
\usetikzlibrary{shapes, arrows, positioning, calc, backgrounds, fit}
\declaretheorem[name=Theorem,refname={Theorem,Theorems},
Refname={Theorem,Theorems},numberwithin=section]{thm}                                
\declaretheorem[name=Lemma,sibling=thm,refname={Lemma,Lemmas},
Refname={Lemma,Lemmas},numberwithin=section]{lem}                  
\declaretheorem[name=Corollary,sibling=thm,refname={Corollary,Corrolaries},
Refname={Corollary,Corrolaries},numberwithin=section]{coro}             
\declaretheorem[name=Definition,refname={Definition,Definitions},
Refname={Definition,Definitions},numberwithin=section]{defn}           
\declaretheorem[name=Proposition,sibling=thm,refname={Proposition,Propositions},
Refname={Proposition,Propositions},numberwithin=section]{prop}          
\declaretheorem[name=Remark,refname={Remark,Remarks},
Refname={Remark,Remarks},numberwithin=section]{remark}                 
\declaretheorem[name=Observation,refname={Observation,Observations},
Refname={Observation,Observations},numberwithin=section]{obs}               
\declaretheorem[name=Fact,refname={Fact,Facts},
Refname={Fact,Facts},numberwithin=section]{fact}                    

\declaretheorem[name=Example,refname={Example,Examples},
Refname={Example,Examples},numberwithin=section]{example}                    

\declaretheorem[name=Open Problem,refname={Open problem,Open problems},
Refname={Open Problem,Open Problems}]{open}                    

\makeatletter
\def\ll@open{%
  \protect\numberline{\csname the\thmt@envname\endcsname}%
  \ifx\@empty\thmt@shortoptarg
    \thmt@open
  \else
    \thmt@shortoptarg
  \fi}
\def\l@thmt@theorem{} 
\makeatother

\usepackage{etoolbox}
\usepackage[disableredefinitions,basic]{complexity}                                     
\usepackage{tikz}

\newcommand{\eps}{\varepsilon}                                      
                
\newcommand{\inprod}[2]{\left\langle #1,#2\right\rangle}            
\newcommand{\from}{:}
\newcommand{\prob}{\operatorname{Pr}}
\newcommand{\size}[1]{\operatorname{size}\left(#1\right)}
\renewcommand{\poly}[1]{\operatorname{poly}\left(#1\right)}

\renewcommand{\polylog}[1]{\operatorname{polylog}\left(#1\right)}
\newcommand{\ftn}{\bF^n_2}
\newcommand{\ztn}{\bZ^n_2}
\newcommand{\st}{\;\vert\;}
\newcommand{\supp}[1]{\operatorname{supp}\left(#1\right)}
\newcommand{\nchi}{\protect\raisebox{2pt}{$\chi$}}

\newcommand{\cA}{\mathcal{A}}
\newcommand{\cB}{\mathcal{B}}
\newcommand{\cC}{\mathcal{C}}
\newcommand{\cD}{\mathcal{D}}

\newcommand{\cF}{\mathcal{F}}
\newcommand{\cG}{\mathcal{G}}
\newcommand{\cH}{\mathcal{H}}

\newcommand{\cN}{\mathcal{N}}
\newcommand{\cO}{\mathcal{O}}

\newcommand{\cQ}{\mathcal{Q}}

\newcommand{\ccS}{\mathcal{S}}
\newcommand{\cT}{\mathcal{T}}
\newcommand{\cU}{\mathcal{U}}

\newcommand{\cW}{\mathcal{W}}
\newcommand{\cX}{\mathcal{X}}
\newcommand{\cY}{\mathcal{Y}}

\newcommand{\bF}{\mathbb{F}}

\newcommand{\bN}{\mathbb{N}}

\newcommand{\bP}{\mathbb{P}}

\newcommand{\bR}{\mathbb{R}}

\newcommand{\bX}{\mathbb{X}}

\newcommand{\bZ}{\mathbb{Z}}

\newcommand{\fD}{\mathfrak{D}}

\newcommand{\fU}{\mathfrak{U}}

\newcommand{\clg}[1]{\left\lceil #1 \right\rceil}

\newcommand{\simpexpect}[1]{\mathbb{E}\,\left[\,#1\,\right]}
\newcommand{\expect}[2]{\underset{#1}{\mathbb{E}}\,\left[\,#2\,\right]}
\newcommand{\err}[2]{\operatorname{err}_{#1}\left(#2\right)}                  \newcommand{\opterr}[2]{\operatorname{opterr}_{#1}\left(#2\right)}            
\newcommand{\simperr}[1]{\operatorname{err}\left(#1\right)}                                \newcommand{\spectrum}[1]{\operatorname{FS}\left(#1\right)}                   
\newcommand{\sign}[1]{\operatorname{sign}\left(#1\right)}     \newcommand{\mem}[1]{\operatorname{MEM}_{#1}}                 

\newcommand{\ex}[2]{\operatorname{EX}\left(#1,#2\right)} 
\newcommand{\qex}[2]{\operatorname{QEX}\left(#1,#2\right)}    
\newcommand{\rex}[3]{\operatorname{REX}^{#3}\left(#1,#2\right)} 
\newcommand{\qrex}[3]{\operatorname{QREX}^{#3}\left(#1,#2\right)}       
\newcommand{\aex}[1]{\operatorname{AEX}\left(#1\right)}         
\newcommand{\qaex}[1]{\operatorname{QAEX}\left(#1\right)}               

\newcommand{\bigO}[1]{{\operatorname{O}\left( #1 \right)}}
\newcommand{\tildeO}[1]{\operatorname{\tilde{O}}\left( #1 \right)}

\newcommand{\bigOmega}[1]{{\Omega\left( #1 \right)}}

\newcommand{\bigTheta}[1]{{{\Theta}\left( #1 \right)}}

\usepackage[nolist,nohyperlinks]{acronym}
\begin{acronym}
\acro{DNF}{Disjunctive Normal Form}
\acro{CNF}{Conjunctive Normal Form}
\acro{PAC}{Probably Approximately Correct}
\acro{RCN}{Random Classification Noise}
\acro{HSP}{Hidden Subgroup Problem}
\acro{LWE}{Learning with Errors}
\acro{LPN}{Learning Parity with Noise}
\acro{NCP}{Nearest Codeword Problem}
\acro{SVP}{Shortest Vector Problem}
\acro{DLP}{Discrete Logarithm Problem}
\acro{NISQ}{Noisy Intermediate-Scale Quantum}
\acro{RAM}{Random Access Memory}
\acro{QRAM}{Quantum Random Access Memory}
\acro{QQRAM}{Quantum Random Access for Quantum Memory}
\acro{CQRAM}{Quantum Random Access for Classical Memory}
\acro{QFT}{Quantum Fourier Transform}
\acro{MQ}{Membership Query}
\acro{EX}{Example Query}
\acro{REX}{Random Noise Example Query}
\acro{AEX}{Agnostic Example Query}
\acro{QEX}{Quantum Example Query}
\acro{QREX}{Quantum RCN Example Query}
\acro{QAEX}{Quantum Agnostic Example Query}
\acro{QGL}{Quantum Goldreich Levin}
\end{acronym}

\newcommand{\DNF}{\ac{DNF}\xspace}
\newcommand{\CNF}{\ac{CNF}\xspace}
\newcommand{\RCN}{\ac{RCN}\xspace}
\renewcommand{\PAC}{\ac{PAC}\xspace}
\newcommand{\HSP}{\ac{HSP}\xspace}
\newcommand{\LWE}{\ac{LWE}\xspace}
\newcommand{\LPN}{\ac{LPN}\xspace}
\newcommand{\NCP}{\ac{NCP}\xspace}
\newcommand{\SVP}{\ac{SVP}\xspace}
\newcommand{\DLP}{\ac{DLP}\xspace}

\newcommand{\MQ}{\ac{MQ}\xspace}

\begin{document}

\title[The Quantum Learning Menagerie]{\begin{center} The Quantum Learning Menagerie\\
{\large(A survey on Quantum learning for Classical concepts)}\end{center}}

\author*[1]{\fnm{Sagnik} \sur{Chatterjee}}\email{chatsagnik@gmail.com}

\affil[1]{\orgdiv{School of Theoretical Computer Science (STCS)},
\orgname{TIFR Mumbai}}


\abstract{ 
This paper surveys various results in the field of Quantum Learning theory, specifically focusing on learning quantum-encoded classical concepts in the Probably Approximately Correct (PAC) framework. The cornerstone of this work is the emphasis on query, sample, and time complexity separations between classical and quantum learning that emerge under learning with query access to different labeling oracles. This paper aims to consolidate all known results in the area under the above umbrella and underscore the limits of our understanding by leaving the reader with 23 open problems.
}




\maketitle

\section{Introduction}
Ever since computer scientists co-opted the notion of computational devices leveraging quantum mechanical principles~\cite{Feynman1982,Feynman1985,Deutsch1985}, the overarching goal in the quantum \textit{computing} community has been to demonstrate speedups in computational tasks over their classical counterparts. This came to a head when the seminal results of Grover~\cite{grover1996fast} and Shor~\cite{shor} opened the proverbial floodgates for research into computational speedups in various areas of traditional computer science, previously thought to have reached their computational barriers.

Learning theory~\cite{vapnikcherovnekis,vapnik74theory,vapnik82,valiant84learnability} is the study of how algorithms can generalize from data. This subfield of theoretical computer science is closely connected to fields like cryptography, coding theory, and circuit complexity, which makes it an extremely interesting area of research. The close of the twentieth century was marked with a series of results on classical impossibility barriers to \textit{efficient learning},\footnote{The notion of learnability and efficient learnability will be formally introduced later on.} even for extremely simple concept classes such as Boolean functions. It was at this critical juncture that \citet{bshouty95dnf} spearheaded the development of the subfield of quantum learning theory by circumventing known classical barriers for learning DNF formulas by using \textit{quantum query algorithms}. This survey aims to capture a vast portion of the literature on quantum learning theory, \textit{subject to a few constraints}.

\subsection{Structure and Organisation}
In \cref{sec:background} we establish basic notation, and provide an overview of Quantum Computing and Computational Learning Theory, followed an introduction of the PAC-learning framework in \cref{sec:prepac}, focusing on the notions of sample complexity (learnability) and time complexity (efficient learnability). We introduce various query models and discuss separations between them in \cref{sec:query}. Of particular importance are \cref{sec:clmem,sec:biasedqueryoracles}, where we highlight separations between sampling and membership queries in the classical and quantum settings, respectively, and propose a new generalized measure of separation between sampling and membership queries in the quantum case. In \cref{sec:qlsamp}, we survey sample complexity separations between classical and quantum PAC learning algorithms. Our discussions move beyond sample complexity separations in the traditional setup considered by previous surveys such as \cite{arunachalam2017survey,Anshu2024survey}, and cover newer research directions that consider separations w.r.t. generalized oracles, separations in the multi-class setting, and separations under white-box access. In \cref{sec:qltime}, we survey computational separations between classical and quantum learning algorithms. We formally introduce the notion of the Hidden Subgroup Problem induced concept classes, and demonstrate that existing separations follow as from separations in this general concept class. We then interpret the seemingly magical quantum learning speedups for \LWE and \LPN as a consequence of learning under powerful new oracle models. We conclude the survey with the latest developments in separations and impossibility results on learning classical Boolean concept classes such as halfspaces, decision trees, juntas, DNFs, and shallow circuits. 

\subsection*{What this survey is not about}
As noted in the abstract, the main goal of this paper is to survey the sample and time complexity-theoretic separations between classical and quantum learning algorithms in the context of learning \textit{classical concept classes}, specifically when restricted to the PAC learning framework. This work can be thought of as an extension to an earlier survey from 2017 by \citet{arunachalam2017survey}, both in terms of giving an overview of new quantum learning results in the last decade and emphasizing the underlying philosophy of learning separations stemming from learning under labeling query oracles. We detail the scope of the survey upfront, before moving on to the main subject matter at hand.

This survey does not touch on many important subareas of quantum learning theory, specifically under different learning models. We highlight two such categories below:
\begin{enumerate}
    \item Quantum learning for quantum objects. This sub-field of quantum learning theory aligns with Feynman's original vision of using quantum computers to figure out quantum phenomena. Examples include \textit{learning} arbitrary quantum states, quantum channels, and quantum processes. We refer the reader to other  surveys~\cite{Dunjko_2018,Anshu2024survey} for a more comprehensive treatment of this area of quantum learning.
    \item Quantum learning under different learning frameworks such as quantum statistical learning, quantum exact learning, etc. We refer the reader to the following manuscript by \citet{nietner2023unifyingquantumstatisticalparametrized} for an introduction to quantum statistical learning, and the following resources~\cite{arunachalam2017survey,Arunachalam2021exact} for more details on quantum exact learning.
\end{enumerate}

\section{Background and Preliminaries}\label{sec:background}
\subsection{Basic Notation}\label{sec:prenotation}
We represent the logarithm to base $2$ by $\log$. For $n\in\bN$, any positive polynomial function in $n$ is denoted by $\poly{n}$. A multivariate polynomial over $n$-variables is denoted as $\poly{x_1,x_2,\ldots,x_n}$.
For $n\in\bN$, we denote by $[n]$ the set $\{1,\ldots,n\}$.

The finite field of order $2$ is denoted by $\bF_2=\{-1,+1\}$ or $\bZ_2=\{0,1\}$, and $\ftn$ (resp. $\ztn$) denotes an $n$-dimensional vector space over $\bF_2$ (resp. $\bZ_2$).
A string $x\in\ftn$ is an $n$-dimensional vector in $\ftn$, where $x_i$ denotes the $i$-th bit of $x$. A Boolean literal is a bit $x_i$ or its negation $\overline{x}_i$. For a string $x\in\ftn$, the string $\overline{x}\in\ftn$ is the string given by $\overline{x}_i=1\oplus x_i$, where $\oplus$ is the bit-wise sum operation. For $a, b\in\ftn$ (or $\ztn$), any bit-wise binary operation $\odot$ results in a string $(a\odot b)\in \ftn$ (or $\ztn$ respectively). The bit-wise AND, OR, and XOR operations are denoted as $a\land b$, $a\lor b$, and $a\oplus b$, respectively. The Hamming weight of $x\in\ftn$ is denoted by $\abs{x}$.

The indicator random variable for an event $E$ is denoted by $\mathbb{I}_{E}$. 
The uniform distribution over a set $\ccS$ is denoted by $\cU\left(\ccS\right)$. The support of a distribution $\fD$ over the Boolean hypercube $\ftn$ is denoted as $\supp{\fD}=\{z\in\ftn\st\fD(z)\neq 0\}$. 
$x\sim\fD$ denotes a string sampled from a probability distribution $\fD$, i.e. $\prob(\mathbf{X}=x)=\fD(x)$. 
Similarly, $x\sim \ccS\iff x\sim\cU(\ccS)$, where $\ccS\subseteq[n]$. 
We denote by $\expect{\fD}{\mathbf{X}}$ the expectation of the random variable $\mathbf{X}$ with respect to the distribution $\fD$, and by $\simpexpect{\mathbf{X}}$ the expectation of $\mathbf{X}$ with respect to the uniform distribution. 
When dealing with inner product spaces, we will sometimes denote $\expect{\fD}{\bX}$ as $\inprod{\fD}{\bX}$. 

A function $f:\ftn\mapsto\bF_2^m$ is known as a \textbf{Boolean function}. Recall that the finite field of order $2$ is denoted by $\bF_2=\{-1,+1\}$ or $\bZ_2=\{0,1\}$, and $\ftn$ (resp. $\ztn$) denotes an $n$-dimensional vector space over $\bF_2$ (resp. $\bZ_2$).
For this survey, we are concerned with decision problems, i.e., binary-valued functions, and therefore mostly consider Boolean functions of the form $f:\ftn\mapsto\bF_2$, unless explicitly stated otherwise. We refer the reader to \cref{sec:booleanappendix} for more details on the representation of Boolean functions.

\subsection{Basics of Quantum Computing}\label{sec:prequantum}

A density matrix $\rho$ corresponding to a $n$-qubit quantum state is a $n\times n$ positive semi-definite (PSD) matrix with trace $1$: $\rho=\sum_{i\in[n]} p_i\ket{\psi_i}\bra{\psi_i}$, where $\{\ket{\psi_i}\}_{i\in[n]}$ are pure states and $\{p_i\}_{i\in[n]}$ form a probability distribution. Density matrices corresponding to pure states have rank $1$, i.e., $\rho=\ket{\psi}\bra{\psi}$. 

\begin{defn}[Trace Distance]\label{def:tracedis}
The trace distance between two quantum states $\rho$ and $\sigma$ is denoted by 
$
D(\rho,\sigma) = \frac{1}{2}\norm{\rho-\sigma}_1.
$
When $\rho$ and $\sigma$ correspond to pure states $\ket{\phi}$ and $\ket{\psi}$ the trace distance is expressed as $
D(\ket{\phi},\ket{\psi}) = \sqrt{1-{\abs{\bra{\phi}\ket{\psi}}}^2}.
$
\end{defn}

\begin{fact}[Distinguishing quantum states]\label{fact:discrimination}
    Suppose an algorithm is tasked with distinguishing between two mixed states $\rho_0$ and $\rho_1$. If the algorithm is given $\rho_b$, where $b$ is picked uniformly at random from $\{0,1\}$, then the best success probability of distinguishing between $\rho_0$ and $\rho_1$ is $\frac{1}{2}+\frac{1}{2}D(\rho_0,\rho_1)= \frac{1}{2}+\frac{1}{2}\sqrt{1-{\abs{\bra{\rho_0}\ket{\rho_1}}}^2}$. 
\end{fact}

Given a Boolean function $f$, a quantum circuit computes $f$ with error $\varepsilon>0$, if $\forall x\in\ftn$, the quantum circuit $Q$ can compute $f(x)$ with probability at least $1-\varepsilon$.\footnote{By compute, we mean that the outcome of the measurement in the ancilla is $f(x)$.} The bounded error query complexity of computing $f$ is the smallest number of calls to any oracle $O_x$ s.t. $Q$ computes $f$ w.p. $\geq 1-\varepsilon$. The bounded error time complexity also takes into account the time taken to implement the unitaries $V_1,\ldots,V_k$. The complexity class $\BQP$ is the class of decision problems efficiently solvable by a quantum computer and is defined as the class of languages $L\subseteq\{0,1\}^*$ that are decidable, with a bounded error probability of at most $\nicefrac{1}{3}$ by a poly-time uniform family of quantum circuits $\{Q_n\}_n$.
\begin{remark}
    Throughout this survey, when we say that a quantum algorithm exists that performs a certain task, we mean that there is a uniform family of quantum query circuits that performs said task.
\end{remark}
\noindent We refer the reader to \cref{sec:quantappendix} for more discussions on quantum algorithms.

\subsection{Overview of Computational Learning Theory}\label{sec:prelearning}

Let $\cX_n\subseteq\ftn$ be a set of $n$-bit instances (or strings), and $\cY$ be the set of labels. For $n\in\bN$, we denote by $\cC=\cup_{n>0}\cC_n$, where $\cC_n=\{c_n\st c_n\from\cX_n\mapsto\cY\}$, a class of functions that maps instances to labels, known as the \textbf{concept class}. 
Similarly, a \textbf{hypothesis class} $\cH=\cup_{n>0}\cH_n$, where $\cH_n=\{h_n\st h_n\from\cX_n\mapsto\cY\}$, is also a class of functions that maps instances to labels. When it is obvious from the context, we drop the subscript $n$ from the definitions. For most of this survey, we assume that the label set is binary. Therefore, the concepts and the hypotheses discussed above are Boolean functions.  

Since multiple classes of Boolean circuits can compute the same Boolean function (see \cref{sec:booleanappendix} for a detailed discussion), the choice of representation of the concept or hypothesis becomes important for computational tasks. We assume that the representation of any concept class or hypothesis class is \textit{fixed}. 
\begin{remark}
    For this survey, the choice of the underlying representation is the concept class $\cC$ / hypothesis class $\cH$ under consideration, and its corresponding size, denoted by $\size{\cC}$ and $\size{\cH}$ respectively, is the maximal cost of the representation of $\cC$ and $\cH$ respectively.
\end{remark} 
 
A learning algorithm (or learner) $A$ is a randomized map from a set of instances $\ccS$,\footnote{The set $\ccS$ is referred to as the training set for $A$.} sampled (i.i.d.) according to an unknown distribution $\fD$ and labeled by an unknown \textit{target} concept $c\in\cC$ to a hypothesis function $h\in\cH$.\footnote{We assume here that both the hypothesis class $\cH$ and the concept class $\cC$ are known, but the specific target concept $c$ is unknown.} Since the label set is binary, most of the algorithms discussed in this survey perform \textit{binary classification}. We note here that the learner $A$ can be either a classical randomized algorithm or a quantum algorithm. Throughout most of this survey, we focus on one metric to quantify the performance of the output hypotheses, namely, the error of a hypothesis. The notion of error semantically captures the number of mistakes made by the hypothesis $h$ with respect to the concept $c$ in different forms.
\begin{defn}[Error]\label{def:error}\sloppy
The error of a hypothesis $h\in\cH:\ftn\mapsto\bF_2$ with respect to an unknown concept $c\in\cC:\ftn\mapsto\{-1,+1\}$ and an unknown underlying distribution $\fD$ is defined as:
    \begin{equation}\label{eq:error}
        \err{\fD,c}{h}:= \expect{x\sim\fD}{\mathbb{I}_{\left(h(x)\neq c(x)\right)}}=\prob_{x\sim\fD}\left[h(x)\neq c(x)\right].
    \end{equation}
\end{defn}
The optimal error of a hypothesis class $\cH$ with respect to an unknown concept  $c\in\cC$ with respect to to an underlying distribution $\fD$ is defined as $\opterr{\fD,c}{\cH}:= \min_{h\in\cH} \err{\fD,c}{h}.$

\begin{remark}\label{rem:ideal}
    A hypothesis with $\simperr{h}=0$ is called the \textbf{ideal hypothesis}.
\end{remark}

\subsection{The PAC-learning framework}\label{sec:prepac}
In this section, we quantify the performance of the learning algorithm $A$ using the \PAC learning model, introduced by \citet{valiant84learnability}.

\begin{defn}[$(\varepsilon,\delta)$-PAC learnability]\label{def:pac}
    A learning algorithm $A$ $(\varepsilon,\delta)$-PAC learns an unknown concept class $\cC$ over $n$-bit instances using hypothesis class $\cH_A$, if for every distribution $\fD$ over $\cX$, and every target concept $c\in\cC$,  there exists a function $m_0\from (0,1)^2\mapsto\bN$, such that for every $\varepsilon,\delta\in[0,1]$, $A$ takes as input a training set $\ccS$ consisting of $m>m_0(\varepsilon,\delta)$ sampled i.i.d. according to $\fD$ and labeled by $c$ and outputs $h\in\cH_A$ s.t.
    \begin{equation}\label{eq:pac}
        \prob\left[\err{\fD,c}{h}\leq\varepsilon\right]\geq 1-\delta.
    \end{equation}
\end{defn}

The probability in \cref{eq:pac} is over the choice of the training set $\ccS$ and any internal randomization of the learner $A$. We also note that in  \cref{def:pac}, there is an implicit \textit{assumption of realizability}, i.e., with probability $1$ over $\ccS$, there exists $h_{\mathrm{opt}}\in\cH$ s.t. $\err{\fD,c}{h_{\mathrm{opt}}}=0$.

\subsubsection{Sample Complexity and Learnability}
In \cref{def:pac}, $m$ is the \textbf{sample complexity} of the learner $A$ for learning $\cC$. If $m=\poly{\nicefrac{1}{\varepsilon},\nicefrac{1}{\delta},n,\size{\cH_A}}$, then the concert class $\cC$ is said to be PAC-learnable by $A$ using the hypothesis class $\cH_A$. A concept class $\cC$ is \textit{learnable} or PAC-learnable using $\cH_A$ if there exists some $\varepsilon,\delta,m_0$ for which \cref{eq:pac} holds.
\begin{remark}
    \cref{def:pac} follows \citet{uml2014}, where the notion of sample complexity is decoupled from the definition of PAC learning. This is done in order to drive home the distinction between learnability and efficient learnability. In other standard textbooks and sources, it is common to find $m=\poly{\nicefrac{1}{\varepsilon},\nicefrac{1}{\delta},n,\size{\cH_A}}$ directly in \cref{def:pac} itself.
\end{remark}
\subsubsection{Time Complexity and Efficient Learnability} If $A$ $(\varepsilon,\delta)$-PAC learns $\cC$ in polynomial time, i.e., in time $\mathrm{poly}{\left(\nicefrac{1}{\varepsilon},\nicefrac{1}{\delta},n,\size{\cH_A}\right)}$, then we say that the concept class $\cC$ is \textbf{efficiently PAC learnable} (or simply, \textbf{efficiently learnable}) by $A$ using the hypothesis class $\cH_A$. For efficient learners, $\size{\cH_A}=\poly{\nicefrac{1}{\varepsilon},\nicefrac{1}{\delta},n,\size{\cC}}.$ We note here that not all \textit{learnable} concept classes with respect to a learner $A$ are \textit{efficiently learnable} with respect to $A$.

\begin{remark}\label{rem:representationchoice}
    All Boolean functions are efficiently PAC learnable if the representation scheme of the underlying concept class is extremely expressive/ verbose, for example, the concept class is described using truth tables. The interesting case is when the underlying concept is represented using Boolean functions, which are believed to have \textit{limited expressivity}, such as Decision trees, since extremely simple Boolean circuits can compute them. We refer the reader to \cref{sec:hypoappendix} for more discussions on hypothesis classes associated with learning algorithms.
\end{remark}

While there are many possible ways to classify PAC learning algorithms, in this survey, we focus on the sample complexity and time complexity of PAC learners (specifically the separation between classical and quantum PAC learners) under various noise models, and with respect to various oracle queries. 

\begin{remark}
One broad classification umbrella not included in this survey deals with the representation and accuracy of the output hypothesis produced by the learner. We refer the reader to \cref{sec:tax1} for a brief discussion on the above taxonomy.
\end{remark}

\section{Query Complexity Separations}\label{sec:query}
\subsection{PAC Learning with queries}\label{sec:prequery}
The learning algorithm $A$ in \cref{def:pac} $(\varepsilon,\delta)$-\PAC learns $\cC$ with respect to random examples (provided in the form of the training set $\ccS$). Alternatively, we can also assume that the learner $A$ can make black-box queries, i.e., the learner has access to an oracle $\cO$. Learning with respect to queries is primarily used to study generalizations or restrictions of PAC learning setups under various forms of labeling noise. For a more detailed discussion on label noise in the PAC framework, the reader is referred to \cref{sec:prenoise}. In this generalized model, we redefine the notion of \PAC learning as follows:

\begin{defn}[\PAC learning with Oracle access]\label{def:pacoracle}
    A learner $A$ $(\varepsilon,\delta)$-\PAC learns concept class $\cC$  over $n$-bit instances, with query access to oracle $\cO$, using hypothesis class $\cH_A$, if for every $\varepsilon,\delta\in(0,1)$, for every distribution $\fD$ over $\cX$, and every target concept $c\in\cC$,  $A$ takes as input a training set $\ccS$ consisting of $m=\poly{\nicefrac{1}{\varepsilon},\nicefrac{1}{\delta},n,\size{\cH_A}}$ sampled i.i.d. according to $\fD$ and labeled by $c$, and makes at most $m$ queries to $\cO$ and outputs $h$ s.t. the error of $h$ is at most $\varepsilon$ with probability at least $1-\delta$.
\end{defn}

\subsubsection{PAC learning with Sampling Oracles}

Sampling oracles are also known as example query oracles.

\subsubsection{Classical Example Query Oracles}\label{sec:orclclassical}

\par\noindent\textbf{The \ac{EX} Oracle:}
Querying the classical random \textbf{example oracle} $\ex{\fD}{c}$ generates labeled examples of the form $(x,c(x))$ where $x\sim\fD$. Using the \ac{EX} oracle, we can redefine the sample complexity and time complexity of \PAC learning concept $c\in\cC$ using $A$ in terms of the number of queries $A$ makes to $\ex{\fD}{c}$ as follows.

\begin{defn}[$\left(\varepsilon,\delta\right)$-PAC learning with access to $\mathrm{EX}$ oracle]\label{def:pacex}
    
    A learner $A$ $(\varepsilon,\delta)$-\PAC learns concept class $\cC$  over $n$-bit instances, with query access to $\ex{\fD}{c}$, using hypothesis class $\cH_A$, if for every $\varepsilon,\delta\in(0,1)$, for every distribution $\fD$ over $\cX$, and every target concept $c\in\cC$,  $A$ makes at most $m=\poly{\nicefrac{1}{\varepsilon},\nicefrac{1}{\delta},n,\size{\cH_A}}$ queries to $\ex{\fD}{c}$ and outputs $h$ s.t. $\prob\left[\err{\fD,c}{h}\leq\varepsilon\right]\geq 1-\delta$.  
\end{defn}

\par\noindent\textbf{The \ac{REX} Oracle:}
Querying an $\rex{\fD}{c}{\eta}$ oracle generates labeled examples of the form $\{(x_i,y^{\prime}_i)\}_{i\in[m]}$, where $y^{\prime}_i=\overline{y_i}$ w.p. $\eta$ and $y^{\prime}_i=y_i$ w.p. $1-\eta$. Hence, \PAC learning in the \RCN setting can be thought of as the learner $A$ making queries to an \ac{REX} oracle.

\par\noindent\textbf{The \ac{AEX} Oracle:}
Similar to the \RCN setting, \PAC learning in the agnostic setting can be thought of as the learner $A$ making queries to an \ac{AEX} oracle, where invoking an $\aex{\fD}$ oracle generates samples $(x,y)\sim\fD$ according to some unknown joint distribution $\fD$ over $\cX\times\cY$. 
The notion of learnability in the agnostic case is defined as follows~\cite{kalai2008agnostic}.
\begin{restatable}[$\beta$-optimal $(\varepsilon,\delta)$-agnostic PAC learning]{defn}{betaoptimal}\label{def:agnosticpacv1}\sloppy
    A learner $A$ $\beta$-optimally $(\varepsilon,\delta)$-agnostic PAC learns a benchmark concept class $\cC$  over $n$-bit instances, with query access to $\aex{\fD}$, using hypothesis class $\cH_A$, if for every $\varepsilon,\delta\in(0,1)$, for some $\beta\in[0,\nicefrac{1}{2})$, for every distribution $\fD$ over $\cX\times\cY$,
    $A$ makes at most $m=\poly{\nicefrac{1}{\varepsilon},\nicefrac{1}{\delta},n,\size{\cH_A}}$ queries to $\aex{\fD}$ and outputs $h$ s.t. $\prob\left[\err{\fD}{h}\leq\opterr{\fD}{\cC}+\beta+\varepsilon\right]\geq 1-\delta$ for some $\beta\in\left[0,\nicefrac{1}{2}\right)$. 
\end{restatable}

Learnability (and efficient learnability) in \cref{def:agnosticpacv1} implies that the number of queries made by $A$ to $\aex{\fD}$ (time taken by $A$ to output $h$, respectively) is polynomially bounded, similar to \cref{def:pacex}. In \cref{def:agnosticpacv1}, if there exists a concept $c\in\cC$ s.t. $\err{\fD}{c}:=\prob_{(x,y)\sim\fD}\left[c(x)\neq y\right]=0,$ then $\opterr{\fD}{\cC}=\beta$, then we are almost back in the noiseless PAC learning setting. 

\begin{obs}
    The concept class $\cC$ serves as a \textit{benchmark} for the learner $A$ in the agnostic setting. If we choose a very bad benchmark concept class, i.e., $\opterr{\fD}{\cC}>\nicefrac{1}{2}$, then the best agnostic learner will produce a hypothesis with accuracy less than $\nicefrac{1}{2}$.
\end{obs} 

\begin{obs}
 In learning under classical sampling oracles, the Learner $A$ obtains copies of a convex combination, where the weights are Probability Masses.
 
 $$\sum_{x\in\mathcal{X}} \fD_x\;\; {(x,c(x))},\;\; \sum_{x\in\mathcal{X}} \fD_x = 1,\;\; \fD:\mathcal{X}\xrightarrow{}[0,1].\; \text{(Noiseless setting.)}$$     
 $$\sum_{(x,y)\in\mathcal{X}\times\mathcal{Y}} \fD_{x,y}\;\; {(x,y)},\;\; \sum_{(x,y)\in\mathcal{X}\times\mathcal{Y}} \fD_{x,y} = 1,\;\; \fD:\mathcal{X}\times\mathcal{Y}\xrightarrow{}[0,1].\; \text{(Agnostic setting.)}$$
    
\end{obs}
\subsubsection{Quantum Example Query Oracles}\label{sec:orclquantum}
\par\noindent\textbf{The \ac{QEX} Oracle:}
\citet{bshouty95dnf} generalized the notion of \nameref{def:pacex} to the quantum setting by introducing the \ac{QEX} oracle (parameterized as $\qex{\fD}{c}$), and defining quantum \PAC learning with access to a \ac{QEX} oracle as follows.
\begin{equation}\label{eq:qexstate}
    \qex{\fD}{c}:\ket{0,0}\mapsto\sum_{x\in\ftn}\sqrt{\fD(x)}\ket{x,c(x)}.
\end{equation}
\par\noindent\textbf{The \ac{QREX} Oracle:}
We can define a quantum example oracle that captures the \RCN setting as follows.
\begin{equation}\label{eq:qrexstate}
    \qrex{\fD}{c}{p}:\ket{0,0}\mapsto\sum_{x\in\ftn}\sqrt{(1-p)\cdot\fD(x)}\ket{x,f(x)}+\sum_{x\in\ftn}\sqrt{p\cdot\fD(x)}\ket{x,\overline{f(x)}}.
\end{equation}

\begin{defn}[Quantum PAC learning with QEX  (resp. QREX) oracle]\label{def:pacqex}
    A quantum algorithm $A$ $(\varepsilon, \delta)$-quantum PAC
    learns concept class $\cC$  over $n$-bit instances, with query access to $\qex{\fD}{c}$ (resp. $\qrex{\fD}{c}{p}$), using hypothesis class $\cH_A$, if for every $\varepsilon,\delta\in(0,1)$, for every distribution $\fD$ over $\cX$, and every target concept $c\in\cC$,  $A$ makes at most $m=\poly{\nicefrac{1}{\varepsilon},\nicefrac{1}{\delta},n,\size{\cH_A}}$ queries to $\qex{\fD}{c}$ (resp. $\qrex{\fD}{c}{p}$) and outputs $h$ s.t. $\prob\left[\err{\fD,c}{h}\leq\varepsilon\right]\geq 1-\delta$.  
\end{defn}

\par\noindent\textbf{The \ac{QAEX} Oracle:}
\citet{arunachalam2017survey} generalized the \acf{AEX} oracle and the notion of \nameref{def:agnosticpacv1} to the quantum setting as follows:
\begin{equation}\label{eq:qaexstate}
\qaex{\fD}:\ket{0,0}\mapsto\sum_{(x,y)\in\bF^{n+1}_2}\sqrt{\fD(x,y)}\ket{x,y}.
\end{equation}
\begin{defn}[Quantum agnostic PAC learning with QAEX oracle]\label{def:pacqaex}
    A quantum learner $A$ $\beta$-optimally $(\varepsilon,\delta)$-PAC 
    learns a benchmark concept class $\cC$  over $n$-bit instances, with query access to $\qaex{\fD}$, using hypothesis class $\cH_A$, if for every $\varepsilon,\delta\in(0,1)$, for some $\beta\in[0,\nicefrac{1}{2})$, for every distribution $\fD$ over $\cX\times\cY$,
    $A$ makes at most $m=\poly{\nicefrac{1}{\varepsilon},\nicefrac{1}{\delta},n,\size{\cH_A}}$ queries to $\qaex{\fD}$ and outputs $h$ s.t. $\prob\left[\err{\fD}{h}\leq\opterr{\fD}{\cC}+\beta+\varepsilon\right]\geq 1-\delta$ for some $\beta\in\left[0,\nicefrac{1}{2}\right)$. 
\end{defn}
\begin{obs}
    If the learner is classical but the query oracles are quantum, the only recourse left to the learner is to measure the state in \cref{eq:qexstate,eq:qaexstate,eq:qrexstate} to obtain one labeled sample. Therefore, the $\ex{\fD}{c}$ and $\qex{\fD}{c}$ oracles (and their noisy versions, respectively) are equivalent to a classical learning algorithm.
\end{obs}
\begin{remark}
    In \cref{def:pacqex,def:pacqaex}, when we say that the quantum learner outputs a hypothesis, we mean that the quantum learner performs a POVM measurement such that each outcome of the measurement is associated with a hypothesis.
\end{remark}
\begin{remark}
    The concepts of learnability and efficient learnability in the quantum PAC setting are similar to those in the classical setting. The notions of complexity that we care about in the quantum PAC setting are, once again, the sample (and/or query) complexity and the time complexity of the learner.
\end{remark}
\begin{remark}
    We use $A^{\cO}$ to denote that learner $A$ makes queries to oracle $\cO$. 
\end{remark}

\begin{obs}
    In learning under quantum sampling oracles, the Learner $A$ obtains copies of a linear combination, where the weights are Probability Amplitudes.
    $$\sum_{x\in\mathcal{X}} \sqrt{\fD_x}\;\; {(x,c(x))},\;\; \sum_{x\in\mathcal{X}} \fD_x = 1,\;\; \fD:\mathcal{X}\xrightarrow{}[0,1].\; \text{(Noiseless setting.)}$$

    $$\sum_{(x,y)\in\mathcal{X}\times\mathcal{Y}} \sqrt{\fD_{x,y}} \;\;{(x,y)},\;\; \sum_{(x,y)\in\mathcal{X}\times\mathcal{Y}} \fD_{x,y} = 1,\;\; \fD:\mathcal{X}\times\mathcal{Y}\xrightarrow{}[0,1].\; \text{(Agnostic setting.)}$$
\end{obs}

\begin{remark}
    Recently, \citet{arunachalam2025efficientlylearningdepth3circuits} introduced a stronger model of quantum agnostic learning under which they obtained state-of-the-art learnability results similar to the ones described in \cref{sec:qltime}. In their model, \citet{arunachalam2025efficientlylearningdepth3circuits} consider that a benchmark quantum phase state that encodes information about the target concept class. The input
\end{remark}
        
\subsubsection{PAC learning with Membership Query Oracles}\label{sec:orclmq}
One of the most important oracles for PAC learning is the Membership Query oracle $\mem{c}$. 
Querying the $\mem{c}$ oracle with an instance $x\in\ftn$ returns the instance's label $c(x)$ directly. As we can see, the \MQ oracle allows us to make very powerful \textit{out-of-distribution} queries. 
\subsection{Power of Classical Membership Queries}\label{sec:clmem}
Learning with the membership queries is an \textit{active} model of learning, where the learner can choose to query information about any instance-label pair it wishes, compared to the \textit{passive} learning models in which the learner can only obtain instance-label pairs that are sampled according to an unknown distribution. We state two results below that give some intuition into the relative strength between the two query models of learning.
\citet{bshouty1995exact} proved the following result.
\begin{lem}[Learnability of Boolean functions]\label{lem:exactlearnability}\sloppy
    Any $n$-bit Boolean function is PAC learnable with Membership queries in polynomial time in $n$ and its \DNF size or its decision tree size.
\end{lem}
Recall that \nameref{def:pacex} is nothing but vanilla \PAC learning, reformulated using query access to \ac{EX} oracle. Hence, \cref{lem:exactlearnability} immediately implies a possible learning separation between \PAC learning with \ac{EX} queries and \MQ queries for \DNF, which is not known to be \PAC learnable with \ac{EX} queries.
The following proposition, which follows from similar arguments by~\cite{arunachalam17sample,apolloni1998sample}, gives an unconditional (albeit polynomial) separation between \MQ oracles and \ac{EX} oracles.

\begin{prop}[\MQ vs \ac{EX} oracles]\label{prop:mqvex}
    Suppose a (publicly known) set of instances $\ccS$ is the largest set shattered by a binary concept class $\cC$, and let $\fD:\ccS\mapsto[0,1]$ be a distribution that puts most of its weight on a particular instance as follows: $$\fD(x) = \begin{cases}
            1-3\varepsilon & \;\;,x=x^\prime, \varepsilon=\nicefrac{1}{\poly{\abs{\ccS}}}\\
            \frac{3\varepsilon}{\abs{\ccS}}, & \;\;,x\in\ccS\setminus\{x^\prime\}.
        \end{cases}$$
    Then, any $(\varepsilon,\delta)$\PAC learner for $c$ with respect to $\fD$ making queries to a $\mem{c}$ oracle has polynomially smaller query complexity than a learner making queries to an $\ex{\fD}{c}$ oracle. 
\end{prop}
\begin{proof}
    \par\noindent\textbf{Learning with respect to \MQ oracle:} 
    Since $\cC$ has VC-dim $\abs{\ccS}$, there exists a concept $c\in\cC$, s.t. $c(x^\prime)=0$, and assigns some labeling in $2^{\abs{\ccS}-1}$ to other elements of the instance space. Let $A$ be a \PAC learner for $c$ with respect to $\fD$ that makes at most $Q_{\text{\MQ}}$ queries to an $\mem{c}$ oracle. $A$ requires $Q=\bigO{\frac{2}{3}\abs{\ccS}}$ queries to \PAC learn $c$ with respect to $\fD$, since we need to learn the labels of at most $\frac{2}{3}\abs{\ccS}$ instances to $\varepsilon$-approximate $c$.

    \par\noindent\textbf{Learning with respect to \ac{EX} oracle:} From Lemma 12 of \citet{arunachalam17sample}, we obtain that $Q_{\text{\ac{EX}}}=\poly{\abs{\ccS}}$ queries, straightforwardly.  This is polynomially larger than the \MQ case.
\end{proof}
Both \cref{lem:exactlearnability,prop:mqvex} show why active learning with \MQ oracles is more powerful than relying on sampling queries for learning. In the next section, we show that there also exists a notion of separation between quantum active queries and quantum sampling queries. Despite the advantages, there have been criticisms of \MQ oracles, outlined in \cref{sec:orclcritism}.

\subsection{Power of Quantum Membership Query Oracles}\label{sec:biasedqueryoracles}
The quantum \MQ oracle $\mathrm{QMEM}(f):\ket{x}\mapsto \ket{x}\ket{f(x)}$ is simply a function oracle ({see~\cref{fig:differentoracles}}), and ubiquitous in designing quantum algorithms such as \cref{alg:FourierSampling}. We now introduce a nomenclature for biased quantum membership query oracles.
\begin{defn}[$\varepsilon$-Biased oracles]\label{def:epsbiasedoracles}
    A quantum example query oracle $\cO^{\varepsilon}_f$ corresponding to a Boolean function $f$ is called an $\varepsilon$-biased for some $0<\varepsilon<\frac{1}{2}$, when 
    \begin{equation*}
        \cO^{\varepsilon}_f:\ket{0}^{\otimes n}\ket{0}\mapsto \sum_{x\in\ftn}\alpha\ket{x}\ket{f(x)}+\sqrt{1-\alpha^2}\ket{x}\ket{\overline{f(x)}},\;\;\;\alpha>\frac{1}{2}+\varepsilon.
    \end{equation*}
\end{defn}
\ac{QREX} oracles is an $\varepsilon$-biased oracle. Recall that an unbiased oracle is an $\varepsilon$-biased oracle with $\alpha=1$.
\citet{iwama2005general} proved the following lemma:
\begin{restatable}[Simulatability of Unbiased Oracles by $\varepsilon$-Biased Oracles]{lem}{iwama}\label{lem:iwama}
    For any $\bigO{T}$ query quantum algorithm that solves a problem with high probability with query access to an unbiased oracle, there exists an $\bigO{\frac{T}{\eps}}$ query quantum algorithm that solves the same problem with high probability with query access to an $\eps$-biased oracle. 
\end{restatable}
The notion of \nameref{def:epsbiasedoracles} was later generalized by \citet{chatterjee2024efficient}, who introduced the following oracle model.
\begin{defn}[Strongly Biased oracles]\label{def:stronglybiasedoracles}
    A quantum example query oracle $\cO^{\varepsilon_x}_f$ corresponding to a Boolean function $f$ is called a strongly biased for some $0<\varepsilon_x<\frac{1}{2},\forall x\in\ftn$, when 
    \begin{equation*}
        \cO^{\varepsilon_x}_f:\ket{0}^{\otimes n}\ket{0}\mapsto \sum_{x\in\ftn}\alpha_x\ket{x}\ket{f(x)}+\sqrt{1-\alpha_x^2}\ket{x}\ket{\overline{f(x)}},\;\;\;\alpha_x>\frac{1}{2}+\varepsilon_x,\forall x\in\ftn.
    \end{equation*}
\end{defn}
Note that learning under the sampling oracle $\ac{QAEX}$ is a stronger model of learning than learning under a strongly biased oracle. To concretely demonstrate that there exists a separation between sampling and active queries even in the quantum setting, we first reprove a result of \citet{bshouty95dnf} showing that \ac{QEX} oracles cannot exactly \textit{simulate} $\mathrm{QMEM}$ oracles w.r.t. the uniform distribution. First, we start by defining the notion of exact simulatability and then extending this to the notion of approximate simulatability. Then, we generalize the result of \citet{bshouty95dnf} to show that \ac{QEX} oracles cannot even approximately \textit{simulate} $\mathrm{QMEM}$ oracles w.r.t. the uniform distribution.

\begin{defn}[Exact Simulatibility of $\mathrm{QMEM}$]\label{def:exactsim}
    A class of quantum (resp. classical) example oracles $\cO$ can exactly simulate $\mathrm{QMEM}$ (resp. \MQ) for a Boolean concept class $\cF$ w.r.t. a distribution $\cD$, if there exists a $\BQP$ (resp. $\BPP$) algorithm $A$ s.t. for all $f\in\cF$ and all $x\in\ftn$, running $A$ with access to $\cO\left(\fD,f\right)$ on $x$ can produce $f(x)$. 
\end{defn}
\begin{thm}[Exact Simulatibility via QEX~\citet{bshouty95dnf}]\label{lem:exactsim}
    \ac{QEX} cannot exactly simulate $\mathrm{QMEM}$ w.r.t. $\cU\left(\ftn\right)$.
\end{thm}
\begin{proof}
    Consider two Boolean functions $f$ and $g$ that differ only in a single string $x^\prime\in\ftn$, e.g., let $f(x)=0$ for all $x\in\ftn$, and $$g(x)=\begin{cases}
        1\;\;,x=x^\prime,\\
        0\;\;,\text{otherwise}.
    \end{cases}$$
    Consider a Boolean function $h$ with the promise that it is either $f$ or $g$. Suppose a quantum learner $A$ makes $Q$ queries to a \ac{QEX} oracle $\qex{\cU}{h}$ and wants to figure out which oracle it is querying. This task is the same as trying to discriminate between the following pair of states:
    \begin{equation}
        \ket{\psi_f}={\left(\sum_{x\in\ftn}\sqrt{\cU(x)}\ket{x,f(x)}\right)}^{\otimes Q}\;\;\;,\;\;\;
        \ket{\psi_g}={\left(\sum_{x\in\ftn}\sqrt{\cU(x)}\ket{x,g(x)}\right)}^{\otimes Q}.\label{eq:statesfordis}
    \end{equation}
    From \cref{fact:discrimination}, we know that the best success probability of distinguishing between $\ket{\psi_f}$ and $\ket{\psi_g}$ is $\frac{1}{2}+\frac{1}{2}D\left(\ket{\psi_f},\ket{\psi_g}\right)$. If this probability is greater than $1-\delta$ for any $\delta\in(0,\nicefrac{1}{2})$, then $\braket{\psi_f}{\psi_g}\leq2\sqrt{\delta\cdot(1-\delta)}$. Therefore, 
    \begin{align}
        \braket{{\psi_f}}{\psi_g}&\leq2\sqrt{\delta\cdot(1-\delta)}\\
        \implies {\left(\sum_{x\neq x^{\prime}} \sqrt{\cU(x)}\right)}^{Q}&\leq2\sqrt{\delta\cdot(1-\delta)}\\
        \implies {\left(1-{\cU(x^\prime)}\right)}^{Q/2}&\leq2\sqrt{\delta\cdot(1-\delta)}.
    \end{align}
    Let $\varepsilon={\cU(x^\prime)}=\frac{1}{2^n}$. Then, we have,
    \begin{align}
        {\left(1-\varepsilon\right)}^{Q/2}&\leq2\sqrt{\delta\cdot(1-\delta)}\\
        \implies {\left(1-\varepsilon\right)}^{Q}&\leq4\delta\\
        \implies Q\log\left(1-\varepsilon\right)&\leq\log\left(4\delta\right)\label{eq:logineqquery1}\\
        \implies \frac{-Q\varepsilon}{1-\varepsilon}&\leq\log\left(4\delta\right)\label{eq:logineqquery2}.
    \end{align}
The inequality in \cref{eq:logineqquery2} follows from applying the inequality $\log(1+x)\geq\nicefrac{x}{1+x},\,\forall x>-1$ to the LHS of \cref{eq:logineqquery1}. Hence, we have that $Q=\bigOmega{\frac{1}{\varepsilon}\log\left(\frac{1}{\delta}\right)}$, i.e., the number of queries is superpolynomial. Therefore, there exists no query-efficient learner $A^{\text{\ac{QEX}}}$ that can distinguish between $\ket{\psi_f}$ and $\ket{\psi_g}$ with prob $\geq 1-\delta$.
\end{proof}
\begin{remark}[Unconditional separation between $\mathrm{QMEM}$ and \ac{QEX} w.r.t. uniform]
    Even though $\mathrm{QMEM}$ cannot be exactly simulated by \ac{QEX} w.r.t. uniform, we can simulate the \ac{QEX} oracle by $1$ query to the $\mathrm{QMEM}$ oracle as follows:
    \begin{equation}
        \mathrm{QMEM}(f)H^{\otimes n}\ket{0}^{\otimes n}\ket{0} = \sum_{x\in\ftn}\frac{1}{2^{n/2}}\ket{x,f(x)}
    \end{equation}
This shows us that at least w.r.t. the uniform distribution, $\mathrm{QMEM}$ queries are unequivocally stronger than \ac{QEX} queries, and therefore, also its noisy counterparts in \ac{QREX} and \ac{QAEX} queries.
\end{remark} 
The proof of \cref{lem:exactsim} shows us that if there exist two functions that differ on exactly one coordinate, no polynomial query quantum algorithm can distinguish the corresponding quantum states. We now formalize this notion to define the concept of approximate simulatability. 

\begin{defn}[$\varepsilon$-far functions]
    A pair of Boolean functions $f,g$ are $\varepsilon$-far if $f\neq g$ only on an $\varepsilon$-fraction of the domain, for any $\varepsilon>0$.
\end{defn}

We are now interested in the following question: Can a learning algorithm with access to (promise) sampling oracles efficiently distinguish between $f$ and $g$ w.h.p., for any $\varepsilon$-far pair $f,g\in\cF$?

\begin{defn}[Approximate Simulatibility of $\mathrm{QMEM}$]
    Consider an arbitrary pair of $\varepsilon$-far Boolean functions $\{f,g\}$, for $\varepsilon>0$, and let $h\in\{f,g\}$.
    For any $\delta>0$, a class of quantum (resp. classical) example oracles $\cO$ can $(\varepsilon,\delta)$-simulate (or approximately simulate) $\mathrm{QMEM}$ (resp. \MQ) for a Boolean concept class $\cF$ w.r.t. a distribution $\fD$, if there exists a BQP (resp. BPP) algorithm $A$ with query access to any $\cO\left(\fD,h\right)$ $A$ can figure out if $h=f$ or if $h=g$, w.p. $\geq 1-\delta$.
\end{defn}

\begin{thm}
    QEX cannot $(\varepsilon,\delta)$-simulate $\mathrm{QMEM}$ w.r.t. $\cU\left(\{0,1\}^n\right)$.
\end{thm}
\begin{proof}
    Suppose $f$ and $g$ are Boolean functions that differ on a polynomial sized subset $\ccS\subset\ftn$, i.e., $\abs{\ccS}=\poly{n}$. Therefore, $\norm{f-g}_1=\frac{1}{2^n}\sum_x\abs{f(x)-g(x)}=\frac{\abs{\ccS}}{2^n}=\nicefrac{\poly{n}}{2^n}=\varepsilon$. Suppose a quantum learner $A$ makes $Q$ queries to a \ac{QEX} oracle $\qex{\cU}{h}$ and wants to discriminate between the states $\ket{\psi_f}$ and $\ket{\psi_g}$ as described in \cref{eq:statesfordis}. From \cref{fact:discrimination}, we know that the best success probability of distinguishing between $\ket{\psi_f}$ and $\ket{\psi_g}$ is $\frac{1}{2}+\frac{1}{2}D\left(\ket{\psi_f},\ket{\psi_g}\right)$. If this probability is greater than $1-\delta$ for any $\delta\in(0,\nicefrac{1}{2})$, then $\braket{\psi_f}{\psi_g}\leq2\sqrt{\delta\cdot(1-\delta)}$. Therefore, 
    \begin{align}
        \braket{{\psi_f}}{\psi_g}&\leq2\sqrt{\delta\cdot(1-\delta)}\\
        \implies {\left(\sum_{x\notin \ccS} \sqrt{\cU(x)}\right)}^{Q}&\leq2\sqrt{\delta\cdot(1-\delta)}\\
        \implies {\left(1-\sum_{x\in\ccS}\sqrt{\cU(x)}\right)}^{Q}&\leq2\sqrt{\delta\cdot(1-\delta)}\\
        \implies {\left(1-\varepsilon\right)}^{Q}&\leq2\sqrt{\delta\cdot(1-\delta)}\\
        \implies {\left(1-\varepsilon\right)}^{2Q}&\leq4\delta\\
        \implies 2Q\log\left(1-\varepsilon\right)&\leq\log\left(4\delta\right)\label{eq:logineqquery3}\\
        \implies \frac{-2Q\varepsilon}{1-\varepsilon}&\leq\log\left(4\delta\right)\label{eq:logineqquery4}.
    \end{align}
    The inequality in \cref{eq:logineqquery4} follows from applying the inequality $\log(1+x)\geq\nicefrac{x}{1+x},\,\forall x>-1$ to the LHS of \cref{eq:logineqquery3}.
    Plugging in the value of $\varepsilon$, we see that $Q=\bigOmega{\frac{2^n}{\poly{n}}\log\left(\frac{1}{\delta}\right)}$. Hence, there does not exist any query efficient quantum learner $A^{\text{\ac{QEX}}}$ that can distinguish between $\ket{\psi_f}$ and $\ket{\psi_g}$ with probability $\geq 1-\delta$.
\end{proof}

\begin{table}[htbp]
    \centering
    \setlength{\tabcolsep}{3pt} 
    \renewcommand{\arraystretch}{1.4}
    \caption{A Taxonomy of Classical and Quantum PAC Learning Oracles. Classical \textbf{Sampling} oracles produce a random sample, while Quantum \textbf{Sampling} oracles produce a quantum state. \textbf{Active} oracles allow arbitrary queries. The representation column contrasts classical bit-pairs with quantum superposition states.}
    \label{tab:oracle_taxonomy}
    \begin{tabular}{@{} l l l l @{}}
        \toprule
        \textbf{Name} & \textbf{Paradigm} & \textbf{Description} & \textbf{Sample/State Obtained} \\ 
        \midrule
        \multicolumn{4}{@{}l}{\textit{\textbf{Classical Oracles}}} \\
        \textbf{EX} & Sampling & Standard Example & $(x, c(x))$ where $x \sim \fD$ \\
        \textbf{REX} & Sampling & R. Class. Noise & $(x, y')$ where $y'$ flipped w.p. $\eta$ \\
        \textbf{AEX} & Sampling & Agnostic Example & $(x, y) \sim \fD_{\cX \times \cY}$ (Joint Dist.) \\
        \textbf{MEM} & Active & Membership Query & Input $x \to$ Returns $c(x)$ \\
        \midrule
        \multicolumn{4}{@{}l}{\textit{\textbf{Quantum Oracles}}} \\
        \textbf{QEX} & Sampling & Quantum Example & $\sum_{x} \sqrt{\fD(x)} \ket{x, c(x)}$ \\
        \textbf{QREX} & Sampling & Quantum RCN & $\sum_{x} (\sqrt{(1-p)\fD(x)}\ket{x, c(x)} + \newline \sqrt{p\fD(x)}\ket{x, \overline{c(x)}})$ \\
        \textbf{QAEX} & Sampling & Quantum Agnostic & $\sum_{(x,y)} \sqrt{\fD(x,y)} \ket{x, y}$ \\
        \textbf{QMEM} & Active & Quantum Mem. & $\ket{x}\ket{0} \mapsto \ket{x}\ket{c(x)}$ (Superposition) \\
        \bottomrule
    \end{tabular}
\end{table}

\subsection{Separations in other Query Models}
To address concerns levied against the \textit{unrestricted} nature of queries the learner can make to an \MQ oracle, \citet{awasthi13} introduced the local \MQ oracle model, where the learner can only query the \MQ oracle on some neighborhood\footnote{Fix some metric over the instance space $\cX$, e.g., consider Hamming distance $d$ when instances are sampled from the Boolean hypercube. Now, we are only allowed to query samples that are within an open ball of radius $d$ around any instance in the training set.} of a training set $\ccS$ provided to the learner. Unfortunately, it was shown by \citet{bary20mq} that even though learning with $q$-local \MQ, for any constant $q>0$ is stronger than vanilla \PAC learning, it is not sufficient to learn relatively simple classes of Boolean functions such as juntas or Decision Trees in a distribution-free setting. A similar oracle model is the Random Walk Query (RWQ) oracle model introduced by \citet{bshouty03randomwalk}, where the first point given to the learner is uniformly sampled from $\ftn$, while every succeeding point is obtained by a uniform random walk over $\ftn$. 
\begin{remark}
    We note here that local \MQ and RWQ oracles can be used to obtain efficient \PAC learners for \DNF with respect to the uniform distribution~\cite{awasthi13,bshouty03randomwalk}.
\end{remark}

\citet{childs03glued} showed that there is an exponential separation between Random Walk Queries and Continuous time Quantum Walk Queries (CQWQ) when performing graph traversal over certain classes of finite graphs. The following questions immediately come to mind:

\begin{open}[Separation between Local MQ and CQWQ query learning]
    Can we find a natural concept class where there is a separation between learning with RWQ (or local MQ) queries and learning with CQWQ queries?
\end{open}

\begin{open}[Simulating QMQ using CQWQ queries]
    Can CQWQ queries exactly or approximately simulate QMQ w.r.t. $\cU\left(\{0,1\}^n\right)$.?
\end{open}

\section{ {Sample Complexity separations}}\label{sec:qlsamp}
In the noiseless setting, ~\citet{blumer89sample} and~\citet{hanneke16optimal} showed that the classical $(\varepsilon,\delta)$-PAC sample complexity of learning a concept class $\cC$ with VC-dimension $d$ is 
$\bigTheta{\frac{d}{\varepsilon}+\frac{\log{\left(\nicefrac{1}{\delta}\right)}}{\varepsilon}}$. In the agnostic setting, ~\citet{vapnik74theory} and~\citet{talagrand94agnostic} showed that the $(\varepsilon,\delta)$-agnostic PAC sample complexity of learning a concept class $\cC$ with VC-dimension $d$ is
$\bigTheta{\frac{d}{\varepsilon^2}+\frac{\log{\left(\nicefrac{1}{\delta}\right)}}{\varepsilon^2}}.$

We recall that having access to $k$ classical samples (or quantum samples) is analogous to making $k$ queries to the \textrm{EX} oracle (respectively, the \textrm{QEX} oracle) in the realizable case, and making $k$ queries to the \textrm{AEX} oracle (respectively, the \textrm{QAEX} oracle) in the agnostic case.
The first results on quantum vs classical sample complexity were given by \citet{gotler} who showed that quantum and classical sample complexities were polynomially related in the noiseless setting as $D=\bigO{nQ}$, for any concept class $\cC$ that is PAC learnable using $D$ queries to the \textrm{EX} oracle and $Q$ queries to the \textrm{QEX} oracle.

This result was later improved by \citet{arunachalam17sample}, who showed that the quantum sample complexity in both the noiseless and agnostic settings is equal to the respective classical sample complexities up to a constant factor, thereby proving that quantum samples are not inherently more powerful than classical samples. 
The proof techniques of \cite{arunachalam17sample} included a state identification argument using ``Pretty Good Measurement" and Fourier analysis. \citet{hadiashar24sample} later reproved the results of \cite{arunachalam17sample} using information-theoretic arguments.

\begin{lem}[Quantum Sample Complexity~\cite{arunachalam17sample,hadiashar24sample}]\label{lem:classicalequivquantumPACsample}
    Given a benchmark concept class $\cC$ with VC-dimension $d$, the $(\varepsilon,\delta)$-quantum PAC sample complexity of learning $\cC$ is
    $
    \bigTheta{\frac{d}{\varepsilon}+\frac{\log{\left(\nicefrac{1}{\delta}\right)}}{\varepsilon}}
    $,
    while the $(\varepsilon,\delta)$-quantum agnostic PAC sample complexity of learning $\cC$ is
    $
    \bigTheta{\frac{d}{\varepsilon^2}+\frac{\log{\left(\nicefrac{1}{\delta}\right)}}{\varepsilon^2}}.
    $
\end{lem}
When marginal distributions over the instances are assumed, however, \cref{lem:classicalequivquantumPACsample} may not hold. For example, \citet{Caro2020} gives quantum sample complexity speedups for Boolean linear functions under bounded noise settings when the samples are drawn from a product distribution. This raises the following question:
\begin{open}[Quantum Sample Complexity separations under non-uniform marginals]\label{open:nonuniformsample}
    Does there exist a natural concept class of non-linear Boolean functions for which we can demonstrate quantum sample complexity speedups over non-uniform marginal distributions?
\end{open}
In the case of non-uniform marginals, one can restrict our attention to special marginals such as product distributions, smoothed product distributions, or isotropic log-concave distributions. We note that even in the case of a negative answer to \cref{open:nonuniformsample}, this still leaves open the possibility for computational separations, as discussed in upcoming sections.

\subsection{Generalized Quantum Example Oracles}\label{sec:mos}
Recently, \citet{caro2024ITCS} proposed a new type of quantum agnostic example oracle, which they term mixture-of-superpositions (MOS) oracle $\cO_{\fD}$. For any arbitrary distribution $\fD$ over $\cX\times\bF_2$, querying $\cO_{\fD}$ yields the mixed state 
\begin{equation}
    \rho_\fD=\expect{f\sim\cF_\fD}{\ket{\psi_{\fD_\cX,f}}\bra{\psi_{\fD_\cX,f}}},
\end{equation}
where $\fD_\cX$ is the marginal distribution over $\cX$, and $\cF_\fD$ is the probability distribution over all labeling functions $f:\cX\mapsto{\{0,1\}}$ induced by $\cD$. Measuring all $n+1$ qubits of $\rho_\fD$ produces a sample from $\fD$ as follows:
\begin{align}
    \bra{x,y}\rho_\fD\ket{x,y} &= \expect{f\sim\cF_\fD}{\abs{\bra{x,y}\ket{\psi_{\fD_\cX,f}}}} =\expect{f\sim\cF_\fD}{\fD_\cX(x)\cdot\mathbb{I}_{f(x)=y}}\\ 
    \implies \bra{x,y}\rho_\fD\ket{x,y}&= \fD_\cX(x)\cdot\underset{f\sim\cF_\fD}{\prob}\mathbb{I}_{f(x)=y}] = \fD(x,y).
\end{align}
The MOS oracle straightforwardly generalizes the $\mathrm{QEX}$, $\mathrm{QREX}$, and $\mathrm{QAEX}$ oracles, and it is admittedly more general than the $\mathrm{QAEX}$ oracle since it allows us to model correlated label noise. Later in \cref{sec:qllattices}, we shall see results that indicate the power of the MOS oracle. However, with respect to the sample complexity, \cite{caro2024ITCS} proves that in the agnostic setting, quantum (MOS) sample complexity is equivalent to classical sample complexity up to logarithmic factors. Concretely, the MOS quantum agnostic sample complexity is $\bigOmega{\frac{d}{\varepsilon^2\log{d}}+\frac{\nicefrac{1}{\delta}}{\varepsilon^2}}$. This naturally leads us to the following question.
\begin{open}[Optimal Mixture-of-Superposition Sample Complexity Lower Bound]\label{open:mos}
    Is the $\frac{d}{\varepsilon^2\log{d}}$ term in the mixture-of-superpositions sample complexity tight? 
\end{open}
If there is a positive answer to \cref{open:mos}, it would indicate that the MOS model provides a slight (but non-trivial) advantage over the standard quantum example models even in terms of sample complexity.

\subsection{Sample Complexity in the Multi-class setting}
\citet{mohan23sample} investigated a host of different learning settings, such as noisy and agnostic offline and online quantum PAC learning, where $\abs{\cY}<\infty$, and showed that quantum and classical sample complexities are equivalent up to constant factors in all but one of these cases (the bounds in the online multi-class agnostic setting are not tight and differ from the best-known classical bounds by a log factor\footnote{See section 5.4 of \cite{mohan23sample} for a detailed discussion on this setting}). The results of \citet{mohan23sample} naturally lead to the following questions.

\begin{open}[Multi-class Agnostic Sample Complexity Separation]
    Are the classical and quantum query/sample complexities equivalent (up to constant factors) in the online multi-class agnostic setting?
\end{open}
\begin{open}[Sample Complexity Separation for Unbounded Labels]
    Are the classical and quantum query/sample complexities equivalent (up to constant factors) in the unbounded label case, even in the realizable offline PAC setting?
\end{open}
\begin{open}[Multi-class Mixture-of-Superposition Sample Complexity]
    What is the mixture-of-superposition sample complexity~\cite{caro2024ITCS} in the multi-class realizable and agnostic settings?
\end{open}
Finally, we may ask the above questions in terms of other recently introduced quantum online learning frameworks as well.
\begin{open}[Sample Complexity separations in Different Frameworks]
    Does the online learning framework of \citet{mohan23sample} extend to the more general quantum online-learning framework introduced by \citet{bansal2024onlinelearningpanoplyquantum}?  
\end{open}

\subsection{Sample Complexity separations via Unitary Access} 
In \cref{def:pacqex,def:pacqaex}, instead of assuming query access to the $\mathrm{QEX}$ and $\mathrm{QAEX}$ oracles, the learner can instead be provided with access to copies of the corresponding quantum states in the RHS of \cref{eq:qexstate,eq:qaexstate}. The results of \cite{arunachalam17sample,hadiashar24sample,mohan23sample} also hold in this equivalent setting. The learning algorithms throughout this survey assume this setup.

On the other hand, if we assume a white-box setting, where the quantum learner has unitary access to the quantum oracles generating the quantum samples directly (instead of only being able to access the quantum samples), then \citet{salmon24sample} showed that the sample complexity of $(\eps,\delta)$- (noiseless) PAC learning a concept class $\cC$ with VC-dimension $d$ is $\tildeO{\frac{d}{\sqrt{\varepsilon}}+\frac{\log{\left(\nicefrac{1}{\delta}\right)}}{\sqrt{\varepsilon}}}.$ 
\begin{open}[White-Box Agnostic Sample Complexity Separations]
    Does the quantum sample complexity advantage of \citet{salmon24sample} extend to the agnostic setting?
\end{open}
The proof of \citet{salmon24sample} critically uses the fact that having unitary access implies having access to the inverses of the unitaries. One might ask the following question, in terms of a more grey-box setting.   
\begin{open}[Grey-Box Sample Complexity Separations]
    What is the quantum (realizable/agnostic) sample complexity of learning when the quantum learner has access to a quantum channel that generates quantum examples?
\end{open}
Here, we assume that having access to a quantum channel does not necessarily imply access to its inverse. This question becomes very interesting in light of a recent work by \citet{tang2025amplitudeamplificationestimationrequire}, who show that access to the inverse is a crucial component in quantum speedups involving amplification and estimation. Finally, we re-examine the sample complexity results of \citet{mohan23sample} under this new approach. 
\begin{open}[White-Box Multi-class Sample Complexity Separations]
    Does the quantum sample complexity advantage of \citet{salmon24sample} extend to the multi-class setting as explored by \citet{mohan23sample}?
\end{open}

\section{{ Time Complexity separations}}\label{sec:qltime}
The results in \cref{sec:qlsamp} are information-theoretic in nature and do not disallow quantum vs classical separations in the computational sense. In other words, the results in \cref{sec:qlsamp} indicate that every concept class $\cC$ that is learnable by quantum learners using a polynomial number of samples/queries is also learnable by classical learners using a polynomial number of samples/queries. However, \cref{lem:classicalequivquantumPACsample} does not imply that if a concept class $\cC$ is efficiently learnable by a quantum learner, it must be efficiently learnable by a classical learner. In this section, we investigate separations in the running times between quantum and classical learners.

\subsection{Separations via the Hidden Subgroup Problem}
\citet{gotler} first provided strong evidence that there is a separation between quantum and classical learning by using a concept class $\cC$ based on factoring Blum integers\footnote{A Blum integer is an integer $N=pq$ where $p,q$ are $m$-bit primes congruent to $3\mod{4}$.}~\cite{KV94blumintegers}. Factoring Blum integers is believed to be computationally intractable classically~\cite{KV94blumintegers,angluin95mem}, but can be efficiently solved by quantum algorithms for the Abelian Hidden Subgroup Problem (Abelian \HSP)~\cite{shor}. 

A similar separation between classical and quantum was later demonstrated by \citet{Liu2021}, using a concept class $\cC_{\text{DLP}}$ based on the \DLP. 
\begin{defn}[\DLP]
Given a large prime $p$, a generator $g$ of $\bZ^*_p=\{1,2,\ldots,p-1\}$, and an input $x\in\bZ^*_p$, compute $\log_g x$ in time $\poly{\clg{\log_2 p}}$.    
\end{defn}
It is widely believed that the \DLP problem is classically intractable~\cite{blum84crypto}, but we can efficiently solve the \DLP problem using Shor's algorithm~\cite{shor}. Using this technique, \cite{Liu2021} obtained a quantum kernel estimation procedure to learn $\cC_{\DLP}$. 
\subsubsection{The Hidden Subgroup Problem over finite Abelian groups}
The results obtained by \cite{gotler,Liu2021} are not surprising, especially in light of Simon's algorithm and Shor's algorithm, which are textbook examples of quantum speedups. Both Simon's algorithm and Shor's algorithm solve restricted instances of a more general problem known as the \textrm{Hidden Subgroup Problem} (\HSP), defined below.
\begin{defn}[The Hidden Subgroup Problem]\label{def:hsp}
    Let $\cG$ be a known finite group, $\ccS$ be a finite set, and a black-box function $f:\cG\mapsto\ccS$ s.t.
    $
    f(x)=f(y)\iff x\cH=y\cH,
    $
    where $\cH\leq\cG$ is some unknown subgroup, and $h\in\cH$. Determine a generating set for $\cH$.
\end{defn}

In \cref{def:hsp}, the black-box function $f$ is said to \textit{hide} the subgroup $\cH$. Alternatively, $f$ can be thought of as separating the cosets of $\cH$. Most exponential separations between quantum and classical computing are a consequence of quantum speedups with respect to \HSP over Abelian groups. 
\begin{example}
    Restricted versions of the Abelian-\HSP problem amenable to quantum speedups include 
    \begin{itemize}
        \item \textbf{Deutsch's problem:} $\cG=\bZ_2$ and $\cH$ is either $0$ (balanced) or $\bZ_2$ (constant),
        \item \textbf{Simon's problem:}  $\cG=\ztn$ and $\cH=\{0,s\}$, where $s\in\ztn$ is some secret, and
        \item \DLP: $\cG=\bZ_p\times\bZ_p$, and $\cH=\{(0,0),(1,\log_g x),\ldots,(p-1,(p-1)\log_g x)\}$.
    \end{itemize}
\end{example}
In fact, we know how to efficiently solve any Abelian \HSP problem.
\begin{lem}[Efficiently solving Abelian \HSP, Theorem 2.2 of \cite{ettinger2000}]\label{lem:HSPfinite}
    Let $\cG$ be a known finite group, $\ccS$ be a finite set, $f:\cG\mapsto\ccS$ be a black-box function that hides a subgroup $\cH\leq\cG$. There exists a quantum algorithm that has size $\bigO{\polylog{\abs{\cG}}}$ and runs in time $\log{\abs{\cG}}$ to output a generating set of $\cH$ w.p. at least $1-\nicefrac{1}{\abs{\cG}}$.
\end{lem}

A critical component in efficiently solving the \HSP problem over finite Abelian groups is the existence of an efficient Quantum Fourier Transform (QFT) over $\cG$. 
\citet{fastqft} showed the existence of a $\bigO{\log n+\polylog{\nicefrac{1}{\varepsilon}}}$ depth and $\bigO{n\log{\nicefrac{n}{\varepsilon}}}$ size circuit to approximate QFT over the group $\bZ_{2^n}$. \citet{cheung01decomposing} proposed an efficient quantum decomposition algorithm, where given a generating set $\{g_1,\ldots,g_s\}$ of a finite Abelian group $\cG$, we can output a set of elements $d_1,\ldots,d_\ell$ s.t. $\cG=\bZ_{d_1}\times\ldots\times\bZ_{d_\ell}$ in time $\polylog{\abs{G}}$. Combining the results of \citet{fastqft,cheung01decomposing} gives us Shor's factoring algorithm (via \cref{lem:finiteabelianstructure}).

\subsubsection{HSP as a Concept class}
We now define a new concept class induced by the \HSP problem.
\begin{defn}[\HSP induced Concept class]\label{def:hspconcept}
    Any decision problem over an \HSP instance $\cP=(\cG,\cH)$ that can be posed as a membership problem over a set of instances $\cX$ is said to be a concept class $\cC:\cX\mapsto\{0,1\}$ induced by $\HSP_{\cP}$. 
\end{defn}
\begin{example}
    Let $\cP=(\cG,\cH)$ be an \HSP instance s.t. $\abs{\cG}=2^n$. The following decision problems $\cC:\ftn\mapsto\{0,1\}$ are concept classes induced by $\HSP_{\cP}$: 
    \begin{itemize}
        \item {Given a string $x\in\ftn$, is $\abs{\cH}\geq \abs{x}?$}
        \item Given a string $x\in\ftn$, does $\cH$ have a generating set of size at least $\abs{x}$?
    \end{itemize}
\end{example}
\begin{remark}
    \cref{def:hspconcept} is not restricted to the Abelian \HSP case and can be defined for \HSP over arbitrary finite groups.
\end{remark}

Classically, obtaining efficient learners for concept classes induced by \HSP over arbitrary finite Abelian groups is subject to various hardness assumptions, while concept classes induced by \HSP over finite Abelian groups seem to admit efficient quantum learning algorithms. This leads us to the following open question.
\begin{open}[Separation for Abelian \HSP-induced concept classes]\label{prop:learninghsp}
    If $\cC:\ftn\mapsto\{0,1\}$ is any concept class induced by some \HSP over a finite Abelian group $\cG$, can we show that there exists an efficient quantum learning algorithm for $\cC$? Does the same hold under restricted and unrestricted noise models?
\end{open}
\begin{remark}
    We conjecture that \cref{prop:learninghsp} is true at least in the noiseless case.
\end{remark}
\begin{obs}
    Assuming the hardness of \DLP~\cite{KV94blumintegers,angluin95mem}, \cref{prop:learninghsp} implies the separation between quantum and classical learning for several \HSP-induced concept classes even in the noiseless case (as seen in \cite{gotler,Liu2021}).
\end{obs}

\subsubsection{The Hidden Subgroup Problem over Arbitrary Groups}
Even though known efficient algorithms for \HSP over Abelian groups require efficient subroutines for QFT, it is unknown if the ability to efficiently perform QFT over a group $G$ is either a necessary or a sufficient condition for efficiently solving \HSP over $\cG$ in the non-Abelian case. \citet{ettinger2004} show that there exist QFT based quantum algorithms to solve \HSP over arbitrary finite groups $\cG$, which make only $\log\abs{G}$ queries to $O_f$ but requires exponential running time.

Efficient algorithms exist for exactly calculating or approximating QFT over classes of \textit{well-behaved} non-Abelian groups, such as metacyclic groups\footnote{A metacyclic group $\cG$ is a group having a cyclic normal subgroup $\cN$ s.t. $\cG/\cN$ is also cyclic.}, or even metabelian groups\footnote{A metabelian group $\cG$ is a group having a normal subgroup $\cN$ s.t. $\cN$ and $\cG/\cN$ are both Abelian.}~\cite{moore06generalqft}.
The reader is encouraged to peruse the surveys by \citet{lomont2004hiddensubgroupproblem} and ~\citet{wang2010hiddensubgroupproblem} for a comprehensive overview of \HSP over arbitrary finite groups, which is an important problem in its own right, even outside the context of learning theory. Efficient quantum algorithms for \HSP over certain classes of finite non-abelian groups imply very interesting results in theoretical computer science, such as efficient solvability of \textrm{Graph Isomorphism} (\texttt{GI}) and \textrm{Shortest Vector Problem} (\SVP). \SVP, in particular, is central to the field of post-quantum cryptography, and its connection to learning is discussed in the next section.
\subsubsection{Other open questions}
Consider the symmetric group $S_n$. Since $S_n$ is a group with $n!$ elements, defining any \HSP-induced concept class that displays a separation between quantum and classical learnability would be very interesting. It might also be worthwhile looking at concept classes induced by subgroups of $S_n$ with more structure in pursuit of efficient quantum learning algorithms. 

\begin{open}[Separations for non-abelian \HSP-induced concept classes]\label{open:symmetric}
    Does there exist a concept class induced by \HSP on $S_n$ or any subgroup of $S_n$ (such as a permutation group or an alternating group), such that there is a super-polynomial separation between quantum and classical learning algorithms?
\end{open}

We note here that the answer to \cref{open:symmetric} is likely to be negative when we require the additional notion of efficient learnability for quantum algorithms. We specifically highlight the important case of dihedral groups (see \cref{dihedral}) due to their close connection to Lattice-based cryptography. 

Informally, an efficient quantum algorithm for solving dihedral \HSP implies that $\SVP\in\BQP$~\cite{oded09lwe}. This would essentially imply that $\NP\subseteq\BQP$ since \SVP has been shown to be $\NP$-hard under restricted settings~\cite{bennet2023svpsurvey}. The best-known algorithm for solving the dihedral \HSP is a subexponential time algorithm by \citet{kuperberg05}. The following open question underscores the above discussion.

\begin{open}[Hardness of learning Dihedral \HSP-induced concept classes]\label{open:dihedral}
    Can we give hardness guarantees for efficient quantum learning of any concept class induced by \HSP on an arbitrary dihedral group $\cG$ for $N\geq3$? 
\end{open}
\subsection{Separations based on Lattice-Based Problems}\label{sec:qllattices}
In computational hardness, we usually care about problems with known \textit{worst-case hardness guarantees} since we want to design algorithms that run efficiently even on the worst possible input. However, for cryptographic schemes to be secure, we require hardness guarantees to hold even for random keys, i.e., we require \textit{average-case hardness guarantees}. This is a challenging task since it is not immediately apparent how to create hard instances of problems even when it has worst-case hardness guarantees. In fact, for many $\NP$-hard problems like Graph coloring (see~\cite{dyer87}), such a guarantee does not exist.

\subsubsection{Introducing the Learning with Errors Problem}
\citet{ajtai96} showed that for certain lattice-based problems, such a reduction from worst-case to average-case hardness exists.  Following this, \citet{oded09lwe} introduced the \textrm{Learning with Errors} (\LWE) problem and showed that an efficient algorithm to solve the \LWE problem gives us an efficient algorithm\footnote{Efficiently solving \LWE actually implies an efficient solution to the Gap-\SVP problem.} for \SVP. We define the decision version of the \LWE problem now.
\begin{defn}[Decision Version of \LWE ]\label{def:LWE}
    We are given $m$ independent samples $(x,y)\in\bZ^n_q\times\bZ_q$, for some prime $q\geq2$, with the promise that they are sampled according to one of the following distributions:
    \begin{enumerate}
        \item The uniform distribution: $(x,y)\sim \fU\left(\bZ^n_q\times\bZ_q\right)$.
        \item The \LWE distribution: The marginal distribution of $x\in\bZ^n_q$ is uniform and $y=\inprod{s}{x}+\varepsilon\mod q$, where $s\in\bZ^n_q$ is a secret string, and $\varepsilon\sim{\fD}$ is some noise parameter.
    \end{enumerate}
    Determine if the samples are drawn from the \LWE distribution.
\end{defn}

The \LWE problem is also a generalization of another important learning problem - \textrm{Learning Parities with Noise} (\LPN) where $q=2$ and $\mathcal{\fD}$ is a Bernoulli distribution over $\bZ_2$. \LPN is an average-case version of the NP-hard problem decoding a linear code~\cite{lyubashevsky2005parity}. 
\begin{obs}
Without the presence of the random noise term $\varepsilon$, \LWE and \LPN can be efficiently solved using Gaussian elimination over the samples.   
\end{obs}

\begin{remark}
    Unless explicitly specified, the \LWE and \LPN problems are assumed to be in the \RCN setting, i.e., $\varepsilon$ is \RCN noise. 
\end{remark}

In the previous section, we saw that \SVP reduces to the Dihedral \HSP, similar to \LWE. We make the connection between these two classes of problems explicit below.

\begin{remark}[Connection between \LWE and the Dihedral \HSP]
\LWE can be reduced (average-worst case reduction) to a stronger variant of the Dihedral \HSP, known as the robust Dihedral \HSP. Since known subexponential (quantum) algorithms for solving dihedral \HSP~\cite{kuperberg05} do not quite work for robust Dihedral \HSP, there do not exist any known subexponential (quantum) learning algorithms for \LWE.
\end{remark}
\subsubsection{Learning algorithms for \textrm{LWE} and \textrm{LPN}}
\citet{blum03lwe} provided the first classical sub-exponential time and sample learning algorithm for \LWE (and hence \LPN) that requires $2^{\bigO{n/\log{n}}}$ time and samples. Interestingly, \citet{ponnu06lpn} showed that \LPN with respect to uniform marginals in the \textit{agnostic setting} reduces to \LPN with respect to uniform marginals in the \RCN setting. Hence, the algorithm by \cite{blum03lwe} also implies a sub-exponential time/sample learning algorithm for \LPN in the agnostic setting (under uniform marginals).

In the quantum setting, however, there exist efficient quantum learning algorithms for \LWE~\cite{grilo2019LWE} and \LPN~\cite{cross2015LPN}, with access to a special oracle that we denote as the $\mathrm{QREX}_{\text{\LWE}}^{s,\varepsilon}$ oracle, and define as follows:
$$
    \mathrm{QREX}_{\text{\LWE}}^{s,\varepsilon}\ket{0}\ket{0}\mapsto\frac{1}{\sqrt{q^n}}\underset{x\in\bF^n_q}{\sum} \ket{x}\ket{\inprod{s}{x}+\varepsilon\mod{q}}\;\;,\varepsilon\sim\fD.
$$
It is important to note that the results of \citet{grilo2019LWE} only hold when the quantum learner is provided oracle access to the $\mathrm{QREX}_{\text{\LWE}}^{s,\varepsilon}$ oracle, unlike in the classical case, where we only have access to $m$ classical labeled examples. 
\begin{remark}
    The efficient learning results of LWE and LPN~\cite{grilo2019LWE,cross2015LPN} also follow from the MOS model~\cite{caro2024ITCS} as outlined in \cref{sec:mos}.
\end{remark}
Despite the presence of powerful oracle models providing separations, it is important to understand the exact avenue of speedup for quantum learners. The following question therefore naturally comes to mind:
\begin{open}[Quantum learning for \LWE under random examples]\label{open:lwe}
    Given access to quantum examples of the form: 
    $
    \frac{1}{\sqrt{\abs{\ccS}}}\underset{x\in\ccS\subseteq[m]}{\sum} \ket{x}\ket{\inprod{s}{x}+\varepsilon\mod{q}},
    $
    where $\ccS$ is unknown and $x\sim\fU\left(\bZ^n_q\right)$, does there exist an efficient quantum learning algorithm for \LWE? 
\end{open}

We note that it is likely that the answer \cref{open:lwe} to is a negative one, since \LWE is believed to be hard for many reasons. If efficient quantum algorithms exist for \LWE under random examples, then there would exist efficient quantum algorithms for breaking the \LWE cryptosystem. This is highly unlikely since such a quantum algorithm would also imply an efficient algorithm for \SVP. If \SVP is believed to be $\NP$-complete, a positive answer to \cref{open:lwe} would imply $\NP\subseteq\BQP$. 

Unlike \LWE, cryptographic applications for \LPN have been limited in scope~\cite{pietrzak2012LPN}. Recently, a series of works \cite{brakerski2019LPN,Yu2021LPN} showed that the worst-case hardness of the \NCP under very low noise rates is implied by the hardness of \LPN at very high noise rates. However, unlike the hardness guarantees of \SVP, \NCP is not believed to be $\NP$-hard~\cite{brakerski2019LPN}. Therefore, it is important to question the existence of an efficient quantum algorithm for \LPN as a separate question from \LWE.

\begin{open}[Quantum learning for \LPN under random examples]\label{open:lpn}
    Given access to quantum examples of the form: 
    $$
    \frac{1}{\sqrt{\abs{\ccS}}}\underset{x\in\ccS\subseteq[m]}{\sum} \ket{x}\ket{\inprod{s}{x}+\varepsilon\mod{2}},
    $$
    where $\ccS$ is unknown, $x\sim\fU\left(\ztn\right)$, and $\varepsilon\sim\mathrm{Ber}(\bZ_2)$, does there exist an efficient quantum learning algorithm for \LPN? 
\end{open}

\subsection{Learning separations for Halfspaces}
Halfspaces are a class of Boolean functions that are extremely important to the field of learning algorithms. Boosting algorithms, SVMs, and neural networks all employ some variant of the halfspace learning problem under the hood. A halfspace is a Boolean function $c\in\cC:\bR^d\mapsto\bF_2$ defined as follows: 
$$
c(x)=\sign{\inprod{a}{x}-\theta}=\sign{\sum_{i\in[m]} a_ix_i-\theta},\;\;\;\;a,x,\theta\in\bR^d.
$$

Any learner that takes labeled samples drawn according to a distribution $\fD$ over $\bR^d\times\bF_2$ and outputs a hypothesis $h$ s.t. $\err{\fD,\cC}{h}\leq\opterr{\fD}{\cC}+\beta+\varepsilon$ is a $beta$-optimal learner. The learner is considered to be efficient if it runs in time $\poly{d,\frac{1}{\varepsilon}}$. 

It is easy to see that in the realizable case/noiseless setting, halfspaces can be efficiently learnt using linear programming techniques. Classical efficient learners for halfspaces also exist for bounded noise models such as \RCN~\cite{blum96halfspace}, Massart noise~ \cite{diakonikolas20massart} and Tsybakov noise~\cite{diakonikolas21tsybakov} under log-concave marginal distributions. 

The most interesting case of halfspace learning seems to be beyond bounded noise models. Learning halfspaces in the agnostic setting under the uniform marginal assumption directly implies learning algorithms for \LPN. Hence, halfspace learning seems to be a much harder problem than \LPN. Various hardness results are known for improper~\cite{daniely2015weakhalfspace,tiegel2023halfspaces} as well as proper~\cite{guruswami2009hardness,ponnu06lpn,diakonikolas2021optimality} classical learning of halfspaces in the agnostic setting.

It is important to note that the hardness results are in the random classical example setting (without queries). Earlier, we have seen evidence of quantum and classical learning separations in the case of \LWE~\cite{grilo2019LWE} or \LPN~\cite{cross2015LPN} when the quantum learner has query access to the $\mathrm{QREX}$ oracle. It is worthwhile to ask if such a separation is possible even for halfspace learning.

\begin{open}[Efficient Quantum Agnostic Learning of Halfspaces]\sloppy
    Does there exist a sample/query/time efficient quantum agnostic learner with access to $\mathrm{QAEX}$ queries for halfspaces assuming uniform marginal over the instances?
\end{open}

\subsection{Learning separations for Decision Trees}
The decision tree representation of a Boolean function $f:\ftn\mapsto\bF_2$ is a binary tree such that every path from the root to any leaf corresponds to a truth table row of $f$. The leaf nodes are labeled by $\{0,1\}$ or $\{\pm 1\}$, the internal nodes are labeled by the literals: variables $x_i$ (or their negations), and every internal node has two outgoing edges labeled {true} and {false}.
The number of leaves (or the number of internal nodes) of the decision tree is known as its \textbf{size} while the longest path from the root of the decision tree to any leaf node is known as the \textbf{depth} of the decision tree. 
Since the decision tree representation of a Boolean function is not unique, computationally, we would prefer to represent $f$ in terms of the ``smallest" decision tree. However, \citet{hyafil76} showed that constructing the optimal decision tree corresponding to an arbitrary Boolean function is an $\NP$-complete problem.

The study of efficient learning algorithms for decision trees with polynomially bounded depth was initiated by \citet{ehrenfeucht1989learning}, who gave a classical proper quasi-polynomial time and sample learning algorithm in the noiseless setting. It was shown recently by \citet{koch2023superpolynomial} that the results of \citet{ehrenfeucht1989learning} are essentially optimal (even with MQ access). \citet{km91} in their seminal work circumvented the superpolynomial lower bound for decision tree learning by considering the improper learning setting. Their results were later extended to the agnostic setting~\cite{gopalan08dt,feldman09dt,kalaikanade}. However, all of these algorithms required access to membership queries. \citet{chatterjee2024efficient} gave quantum algorithms for learning decision trees in all noise models that are strictly polynomial in all parameters while only requiring access to sampling query oracles.

\subsubsection{Techniques and Open Questions}
\citet{km91} showed that if a Boolean function $f$ is computed by a decision tree with size $t$, then $\norm{\spectrum{f}}_1\leq t$. Alternatively, if a Boolean function $f$ is computed by a decision tree with at most $\bigO{t}$ nodes, then $\norm{\spectrum{f}}_1=\bigO{t}$. Hence, we can straightforwardly apply \cref{coro:epslearningl1norm} to yield a polynomial-time learning algorithm for polynomial-sized decision trees using \MQ queries. The tricky part is removing the dependence on \MQ or $\mathrm{QMQ}$ queries for efficient learnability. Below, we discuss how to achieve this in the noiseless case, and defer the discussion on learning under more restrictive noise settings to \citet{chatterjee2024efficient}.

\citet{bshouty95dnf} showed how to obtain an approximate Fourier sampling state\footnote{see \cref{lem:fouriersampling}} using $1$ call to the $\qex{\cU}{f}$ oracle  as follows: 
\begin{equation}\label{eq:approximateFS}
        \frac{1}{\sqrt{2}}\sum_{\ccS\subseteq[n]}\hat{f}(\ccS)\ket{\ccS}\ket{1} + \frac{1}{\sqrt{2}}\ket{0\ldots0}\ket{0}.
    \end{equation}
Applying \nameref{lem:AmplitudeAmplification} to the state obtained in \cref{eq:approximateFS}, we can now obtain the Fourier sampling state $\ket{\psi}=\sum_{\ccS\subseteq[n]}\hat{f}(\ccS)\ket{\ccS}$ using an additional $\bigTheta{1}$ queries to the $\qex{\cU}{f}$ oracle. In the spirit of \cref{coro:epslearningl1norm}, it is now straightforward to obtain a \textit{weak learner} for decision trees by sampling from the Fourier distribution, and obtain a corresponding strong learner using (quantum or classical) Boosting algorithms.

As the observant reader must have deduced, learning algorithms for improper decision tree learning (especially in the agnostic setting), both in the classical and quantum settings, assume a uniform marginal over the instances. This raises the following question:
\begin{open}[Efficient Quantum Agnostic Learning of Decision Trees under non-uniform marginals]
    Can we obtain efficient quantum (agnostic) learners for polynomially bounded decision trees with non-uniform marginal over the instances, without access to QMQ?
\end{open}
\subsection{Learning separations for DNFs}\label{sec:dnflearning}
A DNF function $f$ is a disjunction of DNF clauses, which in turn are conjunctions of Boolean literals $\Tilde{x}$. Mathematically,
$$f(x)=C_1(x)\lor C_2(x)\lor\ldots\lor C_m(x),$$
where each clause $C(x)=\land_{i\in\ccS} \Tilde{x}_i$ is defined with respect to some $\ccS\in[n]$. The width of DNF $f$ is the maximum number of literals in any clause of $f$, and the size of DNF $f$ is the total number of clauses of $f$.

When the learner only has access to random examples, \citet{daniely2016dnf} showed that even improperly learning DNFs would imply an efficient algorithm for refuting random $k$-SAT formulas.\footnote{If $\cP\neq\NP$, it is hard to refute $k$-SAT formulas~\cite{haastad2001some}.} Since learning DNFs (in the realizable/RCN settings) reduces to learning \LPN~\cite{ponnu06lpn}, the results of \cite{daniely2016dnf} also give hardness results for classically learning \LPN with random examples. Classically, we know of efficient improper learning algorithms for DNF (depth-2 $\AC^0$ circuits) in the noiseless/\RCN settings under the uniform (or ``near"-uniform) marginal assumption, only when the learner has access to membership queries~\cite{jacksonHarmonic}. In the quantum learning regime, however, \citet{bshouty95dnf} gave efficient improper quantum learning algorithms for DNFs in the realizable and \RCN settings using quantum learners with access to $\mathrm{QEX}$ or $\mathrm{QREX}$ oracles. 

In the non-uniform marginal setting, \citet{feldman2012dnf} showed that an efficient classical learning algorithm for DNFs exists under the membership query model when the instances are sampled from any product distribution. They also show that a DNF learning algorithm exists for random examples only when the instances are sampled from a smoothed product distribution{~\cite{spielman2004smooth}}. \citet{kanade2019dnf} showed that DNFs can be efficiently PAC learned by quantum learners with access to the $\mathrm{QEX}$ when the samples are drawn from a product distribution. 
\subsubsection{Techniques and Open Problems}
\citet{JACKSON1997414} showed that any $s$-term DNF is spectrally concentrated on at most $\nicefrac{1}{2}+\nicefrac{1}{\poly{s}}$ fraction of its domain. 
This allows us to straightforwardly apply \cref{lem:kmreal} to learn DNFs classically in the noiseless setting (with a polynomial blowup in $s$ and the noise parameter for the bounded noise setting). Using techniques similar to the ones described in the decision tree learning section, \citet{bshouty95dnf} got rid of the dependence on \MQ by using quantum sampling oracles. This leaves us with the following open question on learning DNFs in the unrestricted noise settings:

\begin{open}[Efficient Quantum Agnostic Learning of DNFs]\label{open:agnosticdnf}
    Does there exist an efficient agnostic quantum learning algorithm (with access to the $\mathrm{QAEX}$ oracle) for learning DNF formulas over uniform marginals?
\end{open}

We would also like to pose the following open question as a follow-up to the results of \citet{feldman2012dnf} and \citet{kanade2019dnf}.

\begin{open}[Efficient Quantum Agnostic Learning of DNFs under non-uniform marginals]\label{open:dnf}
    Does there exist an efficient agnostic quantum learning algorithm (with access to the $\mathrm{QAEX}$ oracle) for learning DNF formulas when the marginal distribution over the instances is a smoothed product distribution?
\end{open}

\subsection{Learning separations for Juntas}\label{sec:juntas}

A junta is a Boolean function that only depends on a subset of its input variables. More formally, a $k$-junta is a Boolean function on $n$ variables that depends only on an unknown subset of $k$ variables and is invariant under the remaining $n-k$ variables.

\begin{fact}
    Every $k$-junta can be expressed as a decision tree of size $2^k$ or a DNF formula of size $k$. Every decision tree of size $k$ can be expressed as a DNF formula of size $k$.
\end{fact}

Although the problem of learning $k$-juntas straightforwardly reduces to that of learning DNF functions of width $k$, we mention the results related to learning $k$-juntas separately here. Trivially, learning a $k$-junta from random classical examples can be done using $2^k\log{n}$ samples and $\bigO{n^k}$ time (by solving linear equations). 

\citet{mossel2004junta} gave the first non-trivial algorithm for learning $k$-juntas which takes time $\bigO{n^{\nicefrac{\omega k}{\omega+1}}\poly{n}}$, where matrix multiplication can be solved in time $\bigO{n^\omega}$.\footnote{Here, $\omega<2.376$ is the matrix multiplication exponent. Hence \cite{mossel2004junta} takes time $\bigO{n^{\nicefrac{2k}{3}}}$.} This was later improved by \cite{valiantjunta} to $\bigO{n^{0.6k}\poly{n}}$. \citet{valiantjunta} also gave a $\bigO{n^{0.8k}\poly{n,\nicefrac{1}{\varepsilon}}}$ for learning $k$-juntas in the \RCN setting with noise rate $\varepsilon>0$. 

In the quantum setting, the DNF learning algorithm of \cite{bshouty95dnf} implies a $\bigO{2^{6k}\poly{n,\nicefrac{1}{\varepsilon}}}$ time and sample learning algorithm for $k$-juntas using QEX queries. \citet{Atici2007} improved upon this result by giving a $\bigO{\frac{k}{\varepsilon}\log k}$ time and $2^k\log{\nicefrac{1}{\varepsilon}}$ sample learning algorithm using QEX queries. Both of these results also hold in the \RCN setting with noise rate $\eta$ with an additional multiplicative factor of $\poly{\nicefrac{1}{\eta}}$. \citet{chatterjee2024efficient} design the first quantum algorithm that learns $k$-juntas in the agnostic setting. However, the running time of their algorithm depends on $n$ and has far worse dependence on $k$ and $\varepsilon$. This brings us to the following question:
\begin{open}[Separations for Agnostic Learning of Juntas]
    What is the best time and sample complexity for (quantum) learning $k$-juntas in the agnostic setting without using membership queries?
\end{open}

\subsection{Learning separations for Shallow Circuits}

In this section, we focus on learning shallow classical circuits as concept classes. We start by considering the simplest classical circuit - the $\NC^0_d$ class.     An $\NC^0$ circuit is a constant-depth, polynomial-size Boolean circuit on $n$-bits that consists of \textrm{NOT}, and bounded fan-in \textrm{AND} and \textrm{OR} gates. When the circuit has depth $d=\bigO{1}$, we use the notation $\NC^0_d$. Since $\NC^0_d$ circuits have bounded fan-in (say $b=\bigO{1}$) and constant depth (say $d=\bigO{1}$), the output bit\footnote{Recall that we are only considering $1$-bit output Boolean functions, i.e., decision problems as concept classes in this survey.} is influenced by at most $b^d = \bigO{1}$ input variables. Hence, $\NC^0_d$ circuits are $\bigO{b^d}$ juntas! We refer the reader to the earlier discussion on learning juntas in \cref{sec:juntas}, for a more detailed understanding of learning concept classes generated by $\NC^0$ circuits.

We now move the discussion on separations between quantum and classical learning algorithms to learning the $\AC^0$ and $\TC^0$ circuit classes as concepts. These circuit classes are ubiquitous in complexity theory, cryptography, and learning theory. We formally define these circuit classes below.

\begin{defn}[$\AC^0_d$ circuits and $\TC^0_d$ circuits]\label{def:ac0tc0}
    An $\AC^0$ circuit is a constant-depth, polynomial-size Boolean circuit on $n$-bits that consists of \textrm{NOT}, and unbounded fan-in \textrm{AND} and \textrm{OR} gates. A $\TC^0$ circuit is a constant-depth, polynomial-size Boolean circuit on $n$-bits that consists of \textrm{NOT}, and unbounded fan-in \textrm{AND}, \textrm{OR}, and \textrm{MAJ} gates. When the circuit has depth $d=\bigO{1}$, we use the notation $\AC^0_d$ and $\TC^0_d$ respectively.
\end{defn}
\begin{fact}
    For any $d>0$, $\AC^0_d\subset\TC^0_d$.
\end{fact}

In the context of Boolean functions and learning theory, $\AC^0$ and $\TC^0$ circuit classes are quite important. We know that $\mathrm{DNF}\in\AC^0_2$ and intersections/majority of halfspaces can be computed by $\TC^0$ circuits. Hence, efficient learning algorithms for $\AC^0$ and $\TC^0$ circuits would imply efficient learning algorithms for $\mathrm{DNF}$ and intersections of halfspaces. However, \citet{arunachalam2021lwe} showed the following impossibility results on quantum learning of $\AC_0$ and $\TC_0$ circuits. 
\begin{lem}[Impossibility of quantum learning shallow classical circuits~\cite{arunachalam2021lwe}]\label{lem:impossac0}
The following impossibility results hold for learning $\TC_0$ circuits.
    \begin{enumerate}
    \item If there does not exist a $\BQP$ algorithm to solve the \LWE problem with random examples, then no quantum learner can efficiently learn $\TC^0_2$ even with access to membership queries.
    \item If there does not exist a $\BQP$ algorithm to solve the Ring-\LWE problem\footnotemark with random examples, then no quantum learner can efficiently learn $\TC^0$ even with access to membership queries.
\end{enumerate}
    Further, if there does not exist a strongly subexponential time quantum algorithm to solve the Ring-\LWE problem with random examples, then there is no quasi-polynomial time quantum learner for $\AC^0$ even with access to membership queries.
\end{lem}
We note that \cref{lem:impossac0} implies that there are no efficient quantum learning algorithms for learning intersections of halfspaces (this class of Boolean functions is in $\TC^0_2$), even with membership queries, assuming the hardness of \LWE. \footnotetext{See~\cite{rlwe} for a formal definition of the \texttt{RLWE} problem and its applications to cryptography.}

Unlike the strong hardness of learning $\TC^0$ circuits, all hope is not lost in the case of $\AC^0$ circuits. We have seen previously that the class $\AC^0_2$ is much easier to learn quantumly using just $\mathrm{QEX}$ and $\mathrm{QREX}$ queries~\cite{bshouty95dnf}, while efficient algorithms are known for classical learners only with access to membership queries. The best-known classical learning algorithm for learning $\AC^0_3$ circuits (\textrm{AND} of DNFs) with bounded top fan-in stems due to \citet{ding2017ac0} in a distribution-free setting. This leads us to the following questions:
\begin{open}[Separations for learning $\AC^0_d$ circuits]\label{open:ac03}
    Does there exist an efficient quantum learner (with access to queries) for $\AC^0_3$ circuits (with bounded top fan-in) under uniform marginals? Alternatively, what is the smallest $d>1$ for which there does not exist an efficient quantum learner (with access to queries) for $\AC^0_d$ circuits (with bounded top fan-in) under uniform marginals? 
\end{open}

\begin{remark}
Although a natural progression to this section is to focus on learning \textit{quantum} counterparts of shallow circuits, we note that such a discussion does not fall under the purview of this survey.
\end{remark}

\section{Conclusion}
Upon closer examination of the numerous results surveyed above, a pattern begins to emerge. For many inherently classical problems, quantum oracles/quantum examples seem to capture the essence of the learning task more naturally compared to their corresponding classical counterparts. We motivate the above statement with the following open question.
\begin{open}[Sampling oracle vs Random Examples in the Classical setting]
    Does there exist a learning task where classical learning algorithms can demonstrate a separation between access to a \textit{succinct representation} of the distribution (for example, in the form of access to sampling oracles) over learning with random examples?
\end{open}
In fact, at first glance, it seems to be extremely difficult to formally distinguish between these two models in the context of classical PAC learning. On the other hand, this is a perfectly natural question to ask for quantum learning algorithms, since quantum algorithms naturally manipulate the succinct representations before the measurement/sampling step. Answering this question for classical learning problems in either direction represents an interesting step towards future research in comparing the relative power of learning oracles and their separations. A negative answer to the above question (the more likely scenario), cements the theoretical advantages of quantum algorithms with respect to learning quantum-encoded classical functions.

\listoftheorems[ignoreall,show={open},title={List of Open Problems}]

\section*{Acknowledgements}
The author thanks Debajyoti Bera, Ambuj Tewari, Piyush Srivastava, and Min Hsiu-Hsieh for providing detailed comments on an earlier draft of this manuscript.
\bibliography{refs}
\begin{appendices}
\section{Quantum Computing}\label{sec:quantappendix}
Classical Boolean circuits are finite directed acyclic graphs with \textrm{And}, \textrm{Or}, and \textrm{Not} gates which take $n$-input bits and outputs $m$-output bits to compute some function $f\in\ztn\mapsto\bZ^m_2$.
Classical circuits take $n$-bits as input and output $m$-bits. A quantum circuit is a generalization of a classical circuit where the input is an $n$-qubit state, the output is an $m$-qubit state, and the \textrm{And}, \textrm{Or}, and \textrm{Not} gates are replaced by \textit{quantum gates}. In this survey, we consider the following augmented (non-minimal) universal gate set as the canonical quantum basis gates: $\left\{X,Y,Z,H,CNOT,S,T\right\}$. 
Sometimes, in quantum circuits, we assume that we are given black-box access to a unitary operator (also known as an \textbf{oracle}), which encodes some information only accessible by queries. See \cref{fig:differentoracles} for examples of oracles.
\begin{figure}[htp]
    \centering
    \begin{tabular}{ll}
     \textbf{Phase Oracles.}& $O^{\pm}_{x}:\ket{i,a}\mapsto{(-1)}^{a.x_i}\ket{i,a}$\\&\\
     \textbf{Function Oracles.}&$O_{f}:\ket{x,a}\mapsto \ket{x,a\oplus f(x)}$\\&\\
     \textbf{Probability Oracles.}&$O_{p}:\ket{x,0,0}\mapsto \sqrt{p_x}\ket{x,1,\psi^1_x}+\sqrt{1-p_x}\ket{x,0,\psi^0_x}$
\end{tabular}    
    \caption{Some common oracles used in quantum algorithms}
    \label{fig:differentoracles}
\end{figure}

In the standard quantum query model formally introduced by \citet{beals01query}, a quantum query algorithm has 3 sets of registers - ancillas $\cA$, workspace $\cW$, and query registers $\cQ$ all initialized to $\ket{0}$. A $k$-query quantum query algorithm proceeds by making $k$ rounds of interleaved unitary operations and calls to some quantum query oracle. The algorithm ends by measuring the ancilla register. We now list two important quantum subroutines that have historically provided most of the speedups in quantum learning theory.

Suppose we have access to a unitary $U$ that prepares a quantum state 
$$
\ket{\psi}=\sqrt{\alpha_\mathrm{good}}\ket{\psi_\mathrm{good}}+\sqrt{1-\alpha_\mathrm{good}}\ket{\alpha_{\mathrm{bad}}},
$$
and we wanted to obtain the "good" state $\ket{\psi_\mathrm{good}}$ with high probability. Simply measuring $\ket{\psi}$ would only give us $\ket{\psi_\mathrm{good}}$ with probability ${\abs{\alpha_\mathrm{good}}}^2$. Hence, on expectation it would take us $\nicefrac{1}{{\abs{\alpha_\mathrm{good}}}^2}$ independent measurements to obtain $\ket{\psi_\mathrm{good}}$ with high probability. However, \citet{bhmt} showed that this can be improved using \nameref{lem:AmplitudeAmplification} as stated below.  
\begin{lem}[Amplitude Amplification]
\label{lem:AmplitudeAmplification}\sloppy
 Let $p>0$ be a constant, and $U$ be a unitary operator s.t. $U\ket{0}=\sqrt{p_0}\ket{\phi_0}\ket{0}+\sqrt{1-p_0}\ket{\phi_1}\ket{1}$ for an unknown $p_0\geq p>0$. There exists a quantum algorithm that makes $\Theta(\sqrt{p^{\prime}/p})$ expected number of calls to $U$ and $U^{-1}$ and outputs the state $\ket{\phi_0}$ with a probability $p^{\prime}>0$.
\end{lem}

Let $\ket{\phi_0}=\ket{\psi_\mathrm{good}}$, and $p_0=\abs{\alpha_\mathrm{good}}^2$ in \cref{lem:AmplitudeAmplification}. Then, setting $p^{\prime}=1$ in \cref{lem:AmplitudeAmplification} gives us $\ket{\psi_\mathrm{good}}$ by making $\bigO{\nicefrac{1}{{\abs{\alpha_\mathrm{good}}}}}$ queries to $U$ and $U^{-1}$ in expectation. This gives us a quadratic speedup over the classical case.

One issue with \nameref{lem:AmplitudeAmplification} is that we may not know $p_0$ and hence, how many times we need to invoke $U$ and $U^{-1}$ in order to obtain $\ket{\psi_\mathrm{good}}$. We use the following result by \citet{ambainis:LIPIcs.STACS.2008.1333} to obtain an estimate of $p_0$ with relative error.

\begin{lem}[Relative Amplitude Estimation]\sloppy
Given a constant $p>0$, an error parameter $\varepsilon>0$, a constant $k\geq1$, and a unitary $U$ such that $U\ket{0}=\sqrt{p_0}\ket{\phi_0}\ket{0}+\sqrt{1-p_0}\ket{\phi_1}\ket{1}$ where either ${p_0}\geq p$ or $p_0=0$. Then there exists a quantum algorithm that produces an estimate $\Tilde{p}_0$ of the success probability $p_0$ with probability at least $1-\nicefrac{1}{2^k}$ such that $\left|p_0-\Tilde{p}_0\right|\leq \eps\cdot p_0$ when $p_0\geq p$, and makes $
  O\left(\frac{k}{\varepsilon\sqrt{p}}\left(1+\log{\log{\frac{1}{p}}}\right)\right) 
$ calls to $U$ and $U^{-1}$ in expectation.\label{lem:RelativeEstimation}
\end{lem}
Classically, we can use multiplicative Chernoff bounds to perform mean estimation with relative error by sampling $\tildeO{\frac{1}{\eps^2}}$ independent copies of $U\ket{0}$ and setting $p=\bigO{1/m}$, where $\ket{\phi_0}$ is a superposition over $m$ basis states. 
Hence, using \nameref{lem:RelativeEstimation}, we can perform mean estimation with relative error $\eps$ with a quadratic speedup over classical mean estimation techniques in terms of dependence on $\eps$ and $m$.
\section{Boolean functions continued}\label{sec:booleanappendix}
There are many semantically equivalent representations of Boolean functions. In this survey, we are concerned with three main types of representations of Boolean functions, as listed below. 
\begin{itemize}
    \item \textbf{Truth Table}: In this representation, we explicitly list the value of the Boolean function for all of its $2^n$ possible inputs.
    \item \textbf{Boolean Circuits}: A Boolean circuit is a finite \textrm{DAG} with clearly designated input nodes, internal nodes, and an output node. For this survey, every node has a fan-in of at most $2$. The input nodes are the bits of the input string, and have fan-in $0$; the internal nodes consist of rudimentary Boolean functions belonging to some basis set of Boolean functions (eg.  \textrm{AND}s, \textrm{OR}s, and \textrm{NOT}s), and the output node (with fan-out $0$) stores a binary value corresponding to the output of the evaluation of the input bits.
    \item \textbf{Boolean Formulas / Propositional Formulas}: We can express any arbitrary Boolean function $f$ as an algebraic formula consisting of rudimentary Boolean functions belonging to some basis set of Boolean functions, applied to its input bits. Formally, any Boolean circuit where every node has fan-out at most $1$ is a Boolean formula. Some well-known canonical representations of Boolean functions are the \DNF and \CNF representations, defined with respect to  \textrm{OR}s and \textrm{AND}s over the input literals. Unlike Truth tables or Circuits, Boolean formulas correspond to a covering of $f^{-1}(0)$ or $f^{-1}(1)$.  
\end{itemize}
\begin{remark}
    Given a universal basis set $\cB$ for Boolean functions, we can define Boolean circuits such that the internal nodes only contain gates from $\cB$. For example, consider the basis sets $\cB_1=\{\textrm{NAND}\}$ or $\cB_1=\{\textrm{NOR}\}$. Similar to Boolean circuits, we can also write a formula in terms of different universal basis sets over Boolean functions. However, for this survey, the canonical basis set for classical Boolean circuits and formulas is $\{\textrm{AND}, \textrm{OR}, \textrm{NOT}\}$. 
\end{remark}
The representation size of a Boolean function can be defined depending on the underlying choice of representation. 
\begin{example}
    A standard representation choice for a Boolean function $f$ is the depth of the \textit{minimal} binary decision tree representation of $f$. Another standard choice would be the size of the minimal width DNF representation of $f$.
\end{example}

\subsection{Fourier Analysis of Boolean functions}\label{sec:prefourier}
Fourier-analytic techniques constitute a very powerful family of tools for designing learning algorithms in general and particularly for learning over the Boolean hypercube~\cite{LMN93}. In this section, we will limit our discussions to Fourier analysis for $\bF_2$-valued ($\pm 1$-valued) or $\bZ_2$-valued ($\{0,1\}$-valued) functions, but most results will also apply to real-valued functions over the Boolean hypercube. 
Every Boolean function $f(x)$, can be \textit{uniquely} represented as a multilinear polynomial known as the \textbf{Fourier representation} of $f$, as follows.
\begin{equation}\label{eq:fourierexpansion}
    f(x) = \underset{\ccS\subseteq[n]}{\sum} \hat{f}(\ccS) \,\nchi_{\ccS}(x),\;\;f:\ftn\mapsto\bF_2.
\end{equation}
In \cref{eq:fourierexpansion}, $\nchi_{\ccS}(x)=\prod_{i\in\ccS}x_i$ (we define explicitly $\nchi_{\ccS}(\varnothing)=-1$) is known as a \textbf{parity monomial}, and $\hat{f}(\ccS)=\simpexpect{f\cdot\nchi_{\ccS}}$ is known as the \textbf{Fourier coefficient} of $f$ with respect to to the subset of literals $\ccS$. Note that if $f:{\{0,1\}}^n\mapsto\bF_2$, then $\nchi_{\ccS}(x)={(-1)}^{\sum_{i\in\ccS}x_i}$. 
The set of all parity monomials $\bP=\big\{\nchi_{\ccS} \st \ccS\subseteq[n]\big\}$ forms a basis for all real-valued functions on the Boolean hypercube $\ftn$. The vector space spanned by $\bP$ forms an inner product space, i.e., for any $f,g:\ftn\mapsto\bF_2$,
    \begin{equation}
        \inprod{f}{g} = \frac{1}{2^n}\underset{x\in\ftn}{\sum} f(x)\cdot g(x) = \expect{\cU\left(\ftn\right)}{f(x)\cdot g(x)} 
        1-2\underset{\cU\left(\ftn\right)}{\prob}\left[f\neq g\right].
    \end{equation}
The above definition of inner product captures the notion of distance between two Boolean functions, and also shows us that the set of all parity monomials forms an orthonormal basis set, since $\inprod{\nchi_{\ccS}}{\nchi_{\ccS}}=1$ for all $\ccS\subseteq[n]$, and $\inprod{\nchi_{\ccS}}{\nchi_{\cT}}=0$ if $\ccS\neq\cT$. This allows us to obtain the following observation.
\begin{obs}[Fourier spectrum of a Boolean function]
    For any Boolean function $f$, the set of squared Fourier coefficients forms a probability distribution known as the \textbf{Fourier spectrum} or \textbf{Fourier distribution} of $f$, and is denoted by $\spectrum{f}$.
    $$\underset{\ccS\subseteq[n]}{\sum}{{\hat{f}(\ccS)}^2}=1\; f,g:\ftn\mapsto\bF_2.$$
\end{obs}
We now state the \nameref{lem:fouriersampling} lemma by \citet{BV97}.
\begin{lem}[Fourier Sampling]\label{lem:fouriersampling}
    Given access to a function oracle $O_{f}$ (see \cref{fig:differentoracles}), there exists a quantum algorithm (see \cref{alg:FourierSampling}) that produces the state $\sum_{\ccS\in[n]}\hat{f}(\ccS)\ket{\ccS}$, using only $1$ query to $O_f$ and $\bigO{n}$ gates.     
\end{lem}
Unlike in the quantum case, explicitly constructing the distribution $\spectrum{f}$ is harder classically, since any $\hat{f}(\ccS)=\simpexpect{f\cdot\nchi_{\ccS}}$ depends on all $2^n$ values of $f$. The state produced in \nameref{lem:fouriersampling} allows us to sample from $\spectrum{f}$, i.e., sample $\ccS$ with probability ${\hat{f}(\ccS)}^2$. We can use \nameref{lem:fouriersampling} to learn linear functions $f$ using one query to $O_f$~\cite{BV97}, whereas, classically, this takes $\bigO{n}$ queries.
\begin{algorithm}[H]
    \DontPrintSemicolon
    \SetKwComment{Comment}{~$\vartriangleright$~}{}
    \SetKwInOut{Input}{Output}
    \KwIn{Oracle $O_f:\ket{x}\mapsto f(x)\ket{x}$, providing black-box access to $f$.}
    {Start with $\ket{0^n}$ and apply $H^{\otimes n}$ to obtain 
    $
        \ket{+^{\otimes n}}=\frac{1}{\sqrt{2^n}}\sum_{x\in\ftn}\ket{x}.
    $}\;
    {Query $O_f$ to obtain the state $\ket{\psi}=\frac{1}{\sqrt{2^n}}\sum_{x\in\ftn}f(x)\ket{x}.
    $}\;
    {Apply $H^{\otimes n}$ to $\ket{\psi}$.}{\label{line:fourierbasis}}\;
    \caption{The Fourier Sampling Algorithm}
    \label{alg:FourierSampling}
    \KwOut{The state $\sum_{\ccS\in[n]}\hat{f}(\ccS)\ket{\ccS}$.}
\end{algorithm}

\section{Computational Learning Theory}\label{sec:learningappendix}
\subsection{The Hypothesis Class associated with a Learner}\label{sec:hypoappendix} 
An implicit point central to the study of computational aspects of learning theory is that every learning algorithm $A$ has its own associated hypothesis class $\cH_A$. The learner $A$ can only output a hypothesis from $\cH_A$.\footnote{We drop the subscript from the hypothesis class when the choice of the learner is clear.} We also assume here that all hypothesis classes $\cH_A$ are \textbf{polynomially evaluatable} by $A$, i.e., given any $x\in\cX$ and $h\in\cH_A$, if $A$ can compute $h(x)$ in time $\poly{n,\size{\cH_A}}$.

\begin{obs}
    We note that certain learning algorithms are \textit{more expressive} than others, courtesy of their respective associated hypothesis classes. However, a learning algorithm with the most expressive hypothesis class may not always be the best choice since learners with larger hypothesis classes would need more computational resources to find the optimal hypothesis in their associated hypothesis class.
\end{obs}
\begin{obs}
Note that given an arbitrary choice of distribution $\fD$ over the $\cX$ and some unknown labeling concept $c\in\cC$, assuming that a learning algorithm $A$ can output an ideal hypothesis is not always realistic. For example, consider scenarios when $\cH_A$ is extremely large or the concept $c$ does not belong to $\cH_A$.    
\end{obs}
In settings where obtaining an ideal hypothesis is both computationally and mathematically impossible, the main objective of a learning algorithm is to \textit{efficiently} output a hypothesis $h$ such that $h\approx c$ with respect to some accuracy metric such as error.
\subsection{Classifying PAC learners based on output hypotheses}\label{sec:tax1}
\subsubsection{Proper and Improper PAC learning} 
Recall \cref{rem:representationchoice} and the preceding discussions while defining the notions of learnability and efficient learnability. In \cref{def:pac}, the choice of representation of the output hypothesis $h\in\cH_A$ produced by a learning algorithm $A$ need not necessarily be the same as the unknown target concept class $\cC$ it is trying to learn.
If the choice of representation of $\cH_A$ is the same as the choice of representation as $\cC$, i.e., $h\in\cC$, then the learner $A$ is a \textbf{proper \PAC learner} for $\cC$. Otherwise, when $h\notin\cC$, $A$ is considered an \textbf{improper \PAC learner} for $\cC$.\footnote{(Im)Proper learning is also known as representation-(in)dependent learning.}

Similar to the realizable setting, in the agnostic setting, if $A$ outputs a hypothesis from the benchmark class $h\in\cC$, then $A$ is a proper agnostic learner for $\cC$. Otherwise, $A$ is an improper agnostic learner for $\cC$.

\subsubsection{Weak and Strong Learning}
Directly designing a \PAC learner for an unknown target concept with respect to an unknown distribution is often challenging. 
In practice, it is easier to design a \textit{weak} \PAC learner - something that succeeds with $\sim 51\%$ accuracy, and then amplify the success probability of such learners to be arbitrarily close to $100\%$. In other words: Given an ensemble of learning algorithms that are \PAC learners in a \textit{weak} sense, we can combine the above ensemble of \textit{weak} learners using \textbf{ensemble learning algorithms} to obtain \textit{strong} \PAC learners\footnote{In the realizable setting, strong \PAC learners satisfy \nameref{def:pac}.}. Later, we formalize the notions of \textit{weak} and \textit{strong} PAC learners under the various noise models explored earlier.

\subsection{Criticisms of \textrm{MQ} oracles and their variants}\label{sec:orclcritism}
Even though a significant portion of the learning theory literature has been dedicated to constructing \PAC learning algorithms with access to \MQ oracles for various Boolean concept classes, it has not been the subject of much practical work. This is possibly due to the paradoxical nature of \MQ oracles themselves: If it were possible to construct \MQ oracles for some target concept class $\cC$ in practice, there would be no need to construct \PAC learners for $\cC$. On the other hand, sampling query oracles such as \ac{EX}, \ac{REX}, \ac{AEX} (and their quantum counterparts) simply extend the usual notion of supervised learning by reframing the definitions in terms of oracle queries.

\begin{remark}
    An infamous experiment was conducted by \citet{baum1992query} where they tried to learn a linear classifier for handwritten characters and digits using human annotators as \MQ oracles. It was seen that, in addition to making out-of-distribution queries, sometimes the learner tended to make out-of-domain queries as well. This may be noted as anecdotal evidence in support of the hypothesis that \MQ oracles may not be useful in practice.
\end{remark}

\subsection{PAC learning under various Label Noise Models}\label{sec:prenoise}
\subsubsection{The Realizable~/~Noiseless setting}
The learning setup introduced in \cref{sec:prelearning} is called the \textbf{realizable setting} since we assume that there always exists a target concept $c\in\cC$ that \textbf{realizes} the labeling of the instances. We also denote this setting as the \textbf{noiseless setting} since the labels in the training set are available to the learner without change. 

\subsubsection{Bounded Noise setting}
We can generalize the realizable setting to account for more realistic learning scenarios as follows. Let $\ccS=\{(x_1,c(x_1)),\ldots,(x_m,c(x_m))\}$ be the training set used to train a learning algorithm $A$. Suppose an adversary \textit{corrupts} the labels of the training set $\ccS$ before it is obtained by the learner $A$ and used to formulate the hypothesis. It is straightforward to see that the sample complexity and time complexity guarantees of \nameref{def:pac} would have to depend on the noise rate. 

In the \RCN setting introduced by \cite{Angluin1988}, we assume that the labels of the training set are flipped with probability $0\leq\eta<\nicefrac{1}{2}$, i.e., given training set $\ccS={(x_i,y_i)}_{i\in[m]}$, set each $y_i=\overline{y_i}$ independently with probability $\eta$. In this setting, the sample and time complexity of a PAC learner depends on $\poly{\nicefrac{1}{\varepsilon},\nicefrac{1}{\delta},n,\size{\cC},\nicefrac{1}{1-2\eta}}$. Later, we will redefine the \RCN setting in terms of noisy oracles. An example of a more challenging class of bounded noise is \textbf{Massart noise}, where each label $y_i$ is flipped with probability $\eta_i\leq\eta<\nicefrac{1}{2}$.

\subsubsection{Unbounded Noise setting} A far more challenging noise model is the agnostic \PAC setting~ \cite{haussler1992agnostic,kearns92agnostic}. In this noise model, each label $y_i$ is corrupted with probability $\eta_i<\nicefrac{1}{2}$. Here, we note that there is no bound $\eta$ s.t. $\eta_i<\eta,\forall i\in[\abs{\ccS}]$. 
If the rate of corruption of labels is unbounded, it is unclear whether a learning algorithm $A$, which was an efficient learner for $\cC$, even remains a \PAC learner for $\cC$ after the noise has been added to the training set.
In the unbounded noise model, we no longer assume that there is a fixed underlying concept $c\in\cC$ that labels the instances $x\in\cX$, since the training set itself can be inconsistent.\footnote{If, for some $i\in[\abs{\ccS}]$, the noise rate $\eta_i$ is very close to $\nicefrac{1}{2}$, the training set can contain samples $\{x_i,+1\}$ and $\{x_i,-1\}$. Hence, the labeling of instances is no longer considered a function.}\footnote{Inconsistency also arises in the bounded noise model, but can be resolved easily by majority vote.} Rather, we assume that the labeled instances $(x,y)$ are jointly sampled from a distribution $\fD$ over $\cX\times\cY$. Therefore, the learning algorithm is \textit{agnostic} to the choice of the target concept.
\subsection{Fourier analytic techniques for Learning}
The crux of learning techniques for a Boolean function $f$ is as follows: 
\begin{itemize}
    \item Given some information about the Fourier spectrum of $f$, explicitly construct a weak learner for $f$.
    \item Use an ensemble learning algorithm such as Boosting to obtain an approximate strong learner for $f$ w.h.p., i.e., obtain a PAC learner for $f$.
\end{itemize}

The \textbf{Fourier weight} of a Boolean function $f$ at \textbf{degree} $k$ is defined as the sum of squared Fourier coefficients of subsets of Hamming weight exactly equal to $k$, i.e., $\cW^{=k}{f}=\underset{\ccS\subseteq[n],\abs{\ccS}=k}{\sum}{{\hat{f}(\ccS)}^2}.$ We can similarly define the quantities $\cW^{\leq k}{f}$ and $\cW^{\geq k}{f}$. A Boolean function $f$ is $\varepsilon$-spectrally concentrated on a set $\cA\subseteq 2^{[n]}$ if
    $
        \underset{\ccS\subseteq[n],\ccS\notin \cA}{\sum}{{\hat{f}(\ccS)}^2}\leq\varepsilon.
    $
Therefore, a Boolean function $f$ is $\varepsilon$-spectrally concentrated on \textbf{degree} $\leq k$, if  $\cW^{>k}{f}\leq\varepsilon.$
We now outline some important results.

\begin{lem}[Low Degree Learning~\cite{LMN93}]\label{lem:lowdegree}
    If a Boolean function $f$ is $\eps$-concentrated on a set $\cA$, and we know $\cA$ explicitly, then there exists a learning algorithm $A$ that can $(\eps,\delta)$-PAC learn $f$ in time $\poly{\abs{\cA},n,\nicefrac{1}{\eps},\nicefrac{1}{\delta}}$ with query access to a $\ex{\cU}{f}$  oracle. If $f$ is $\eps$-concentrated up to degree $k$, then $A$ can learn $f$ in time $\poly{n^k,\nicefrac{1}{\eps},\nicefrac{1}{\delta}}$ with query access to a $\ex{\cU}{f}$  oracle.
\end{lem}
\begin{remark}
    If $k$ is not a constant, then the learner in \cref{lem:lowdegree} does not yield an \textit{efficient} PAC learner for $f$.
\end{remark}

One shortcoming of \cref{lem:lowdegree} is the requirement that we have explicit knowledge of the subsets of literals on which the Boolean function is spectrally concentrated. We can get around this requirement by upgrading our choice of oracle from the $\ex{\cU}{f}$ oracle to the $\mem{f}$ oracle.

\begin{lem}[Learning from Spectral Concentration~\cite{gl89,km91}]\label{lem:kmreal}
    Let $\eps,\delta>0$, and $f:\ftn\mapsto\bF_2$. If $f$ is $\eps$-concentrated up to degree $k$, $\exists$ a learning algorithm $A$ with query access to a $\mem{f}$ oracle that $(\eps,\delta)$-PAC learns $f$ in time $\poly{n^k,\nicefrac{1}{\eps},\nicefrac{1}{\delta}}$.
\end{lem}
There are two important corollaries of \cref{lem:kmreal}. The first result allows us to improve upon \cref{lem:lowdegree} by learning a function that is spectrally concentrated on a subset $\ccS$, without explicit knowledge of $\ccS$. 
\begin{coro}[Learning from $\eps$-concentrated subset of parity monomials]\label{coro:epslearningsubset}
    Let $\eps,\delta>0$, and $f:\ftn\mapsto\bF_2$. If $f$ is $\eps$-concentrated on a subset $\ccS$, $\exists$ a learning algorithm $A$ with query access to a $\mem{f}$ oracle that $(\eps,\delta)$-PAC learns $f$ in time $\poly{\abs{\ccS},\nicefrac{1}{\eps},\nicefrac{1}{\delta}}$.
\end{coro}
The next corollary helps us define \PAC learners from parity monomials that have large Fourier coefficients.
Let $f$ be a Boolean function. Then, the $\ell_1$ norm of $\spectrum{f}$ is defined as
\begin{equation}\label{eq:l1norm}
\norm{\spectrum{f}}_1=\sum_{\ccS\subseteq[n]} \abs{\hat{f}(\ccS)}.
\end{equation}
\begin{coro}[Learning from parity monomials with large Fourier coefficients]\label{coro:epslearningl1norm}
    Let $\eps,\delta>0$, and $f:\ftn\mapsto\bF_2$. $\exists$ a learning algorithm $A$ with query access to a $\mem{f}$ oracle that $(\eps,\delta)$-PAC learns $f$ in time $\poly{\norm{\spectrum{f}}_1,\nicefrac{1}{\eps},\nicefrac{1}{\delta}}$.
\end{coro}
\begin{obs}\label{obs:l1learn}
    \cref{coro:epslearningl1norm} tells us that if the $\ell_1$ norm of any Boolean function is polynomially bounded, then we can construct a \PAC learner for $f$ using queries to an \MQ oracle.
\end{obs}

\section{Group Theory}
A group is a non-empty set $\cG$ is a set with an associative binary operation $(\cdot):\cG\times\cG\mapsto\cG$, such that there exists an identity element in $\cG$ with respect to $(\cdot)$, and every element of $\cG$ has an inverse element in $\cG$ w.r.t $(\cdot)$. If for all $x,y\in\cG$, $x(\cdot)y=y(\cdot)x$, then $\cG$ is an Abelian group. A generating subset $\cH$ of $\cG$ is a subset s.t. every element of $\cG$ can be expressed as a combination of finitely many elements of $\cH$ and their inverses. If $\cG$ is a cyclic group, it has an element $x$ that serves as the generator of the group. 
Any arbitrary finite Abelian group can be decomposed into cyclic groups as follows:
\begin{lem}[Fundamental Structure Theorem for Finite Abelian Groups]\label{lem:finiteabelianstructure}
    Every Abelian group $\cG$ is isomorphic to a group of the form 
    ${\bZ}_{p^{a_1}_1}\times{\bZ}_{p^{a_2}_2}\times
    \ldots\times
    {\bZ}_{p^{a_k}_k},
    $
    where $p_1,\ldots,p_k$ are prime numbers (not necessarily unique), and the decomposition is unique up to the order in which the factors $\left\{p_1^{a_1},\ldots,p_k^{a_k}\right\}$ are written.
\end{lem}
\begin{remark}
    $\bZ_k$ is a finite Abelian group of order $k$ with respect to the $+_{\mathrm{mod}\,k}$ operation. $\bZ_k$ is also a cyclic group if $k=1,2,4$ or some power of an odd prime.
\end{remark}

\begin{defn}[Dihedral Groups]\label{dihedral}
    A dihedral group $D_{2N}$ is a non-abelian group of order $2N$ that is generated by two elements $x$ and $y$ s.t.
    $$
        x^N=1,\;\;\; y^2=1,\;\;\; yxy=-x.
    $$
    $D_{2N}$ is group of symmetries of a regular polygon of $2N$ sides. The elements $x$ and $y$ are rotation and reflection operations about the vertical axis. 
\end{defn}

\end{appendices}

\end{document}